\documentclass[11pt,draftclsnofoot,journal,onecolumn]{IEEEtran}
\ifCLASSINFOpdf
\else
\fi
\usepackage[cmex10]{amsmath}
\usepackage{multicol}

\usepackage{mathrsfs}
\usepackage{float}
\usepackage{tikz}
\usetikzlibrary{arrows.meta}
\usepackage{rotating}

\usepackage[numbers]{natbib}

\usepackage{booktabs,multirow}
\usepackage{tcolorbox}
\tcbuselibrary{skins,breakable}
\usepackage{xcolor}
\usepackage{lipsum}
\usepackage{longtable}
\usepackage{bm,bbm}%
\usepackage{hyperref}
\usepackage{rotating}
\usepackage{makecell}
\usepackage{amsthm,amsfonts,amssymb}
\newtheorem{definition}{Definition}
\newtheorem{lemma}{Lemma}
\newtheorem{theorem}{Theorem}
\newtheorem{proposition}{Proposition}
\newtheorem{corollary}{Corollary}
\newtheorem{remark}{Remark}
\newtheorem{example}{Example}
\newtheorem{construction}{Construction}

\newcommand{\Img}{\mathrm{Im}}

\usepackage{array}
\begin{document}
\title{Data Protection in Function-Correcting Symbol-Pair Codes: Redundancy Bounds and Protection Profiles}
%
%
%

\author{ 
        Anamika Singh and
        Abhay Kumar Singh
\thanks{ A. Singh and A. K. Singh are with the Department of Mathematics and Computing, Indian Institute of Technology (ISM), Dhanbad, India. email: anamikabhu2103@gmail.com, abhay@iitism.ac.in}}

\date{ }
\maketitle

\begin{abstract}
In several storage systems, including DNA storage and flash memory, errors affect neighbouring symbols jointly, and the Hamming metric does not adequately capture such error patterns. The symbol-pair read channel, introduced by Cassuto and Blaum~\cite{cassuto2011codes}, addresses this by reading consecutive pairs of symbols rather than individual symbols. Motivated by this, we introduce function-correcting symbol-pair codes with data protection (FCSPC-DP),
which guarantee reliable recovery of a desired function of the message while simultaneously protecting the message itself against symbol-pair errors. We derive bounds on the optimal redundancy of such codes and establish a relationship with joint-pair distance matrices. We also give explicit constructions of FCSPC-DP for locally pair-bounded functions and symbol-pair weight functions. We introduce the pair-separation constant of a function, the minimum symbol-pair distance between messages sharing a function value, and show that when it is sufficiently large, data protection requires no additional redundancy: the optimal redundancy coincides with that of the corresponding code without data protection. Considering the symbol-pair analogue of the $\alpha$-distance graph, we introduce two code invariants, the generation profile and the disconnection threshold, and use them to characterise a code's protection properties. Relating the two metrics through these invariants yields upper and lower bounds on the symbol-pair threshold in terms of its Hamming counterpart, both of which are attained. We further extend the classical Plotkin and sphere-packing bounds to this setting.
\end{abstract}

\begin{IEEEkeywords}
Error-correction, Function-correcting codes, Symbol-Pair weight distribution function, Irregular Pair distance codes,
Linear codes, Locally bounded functions, Redundancy bounds.
\end{IEEEkeywords}

%
\IEEEpeerreviewmaketitle

\section{Introduction}
     \IEEEPARstart{I}{n} many communication and storage systems, the receiver may not require the full transmitted message. For example, a monitoring node may need only an aggregate statistic, a classifier may require only a decision, and a distributed computation may need only the value of a prescribed function of the data held at a remote site. This problem can be addressed using classical error-correcting codes, which protect the entire message and, consequently, any function defined on it. However, doing so addresses a more demanding problem than the one at hand. In classical error-correcting codes, codewords must be separated by a certain distance to ensure reliable recovery of the entire message. In contrast, if only the function value needs to be recovered, messages that share the same function value need not be distinguished from one another. Lenz \emph{et al.}~\cite{10132545} formalised this observation by introducing \emph{function-correcting codes}(FCCs), in which a systematic encoding separates codewords whose messages carry distinct function values by a prescribed distance while imposing no requirement on codewords sharing a function value. Exploiting the function's structure in this way can substantially reduce redundancy, and the framework has since been developed in several directions~\cite{Premlal2024,ge2025,ly2025}.

    A second line of work concerns the channel model itself. In high-density storage media, the read mechanism cannot resolve individual symbols; hence, each read returns a pair of consecutive symbols. Cassuto and Blaum~\cite{cassuto2011codes} initiated the study of codes for this \emph{symbol-pair read channel} and showed that the appropriate measure of error in such a channel is the symbol-pair distance rather than the Hamming distance. In the symbol-pair read channel, a single corrupted symbol corrupts two pair observations, so pair distance is never smaller than Hamming distance, but the amount by which it is larger depends on how the errors are arranged along the codeword, since a contiguous burst corrupts fewer distinct pairs than the same number of isolated errors. Pair distance, therefore, depends on the cyclic adjacency of coordinates, and quantities that are invariant under coordinate permutation in the Hamming metric, including the minimum distance itself, are not invariant here. Symbol-pair codes have been studied extensively, and the function-correcting framework has been carried over to this channel by Xia, Liu and Chen~\cite{xia2024function} and to the more general $b$-symbol read channel in~\cite{singh2025,sampath2025}. Beyond overlapping-read models, the framework has been developed for channels matched to the Lee metric \cite{verma2025function,11538231} and for the homogeneous distance over finite rings \cite{liu2026function}, the latter simultaneously generalizing the Hamming and Lee metrics. More recently, FCCs have been studied under the sum-rank metric \cite{kammila2026sumrankmetric}, which unifies the Hamming and rank metrics, and under the Rosenbloom–Tsfasman metric for parallel channels \cite{liu2026RT}. Synchronization errors were addressed in \cite{singh2026insdel}, where FCCs for insertion–deletion channels were introduced.

    Function correction alone, however, leaves the underlying message unprotected. The framework of~\cite{10132545} guarantees only that the function value is recoverable and the two messages with the same function value may be mapped to codewords at distance one, so the data itself may be corrupted beyond recovery. In many settings, this is unacceptable. In a network where different nodes evaluate different functions of the same stored word, a guarantee tailored to one function is of no use to another; in a storage system where certain attributes are more critical than others, one wants strong protection for the critical attribute together with a baseline guarantee for everything else. Motivated by this, Rajput \emph{et al.}~\cite{rajput2025function} introduced \emph{function-correcting codes with data protection}, in which the encoding satisfies two distance requirements simultaneously: a minimum distance between all pairs of distinct codewords, protecting the data, and a larger minimum distance between codewords whose messages carry distinct function values, protecting the function. When the one minimum distance is strictly larger than the other, the code is called \emph{strict}, and it is exactly this case that is of interest: the function receives genuinely stronger protection than the data, and the redundancy spent on the function is not simply redundancy spent on the message. The boundary case where both minimum distances are the same imposes nothing beyond an ordinary error-correcting code.

    The two minimum distance requirements are not merely two instances of the same condition, and the asymmetry between them is what makes the problem interesting. A code that meets the stronger requirement everywhere would meet both, but at the cost of redundancy. The question is whether the partition induced by the function can be exploited to obtain stronger protection across classes while incurring the cost of redundancy associated with the smaller minimum distance. Whether this is possible depends on the interaction between the code and the function, not on either alone, and for several familiar code families, it turns out to be impossible. Rajput \emph{et al.}~\cite{rajput2025function} showed that codes whose distance graph is sufficiently connected cannot provide strictly stronger protection for any nontrivial function, ruling out perfect and MDS codes, and subsequently studied the existence question for linear codes~\cite{existencerajput2026}.

        In this paper, we study function-correcting codes with data protection over the symbol-pair read channel, which we call \emph{function-correcting symbol-pair codes with data protection} (FCSPC-DP). We combine two independent modifications of the classical setting: function correction, which changes the decoding objective, and the symbol-pair channel, which changes both the error model and the distance measure. Symbol-pair distance is sensitive to the cyclic arrangement of coordinates, so pair distances are not additive across a concatenation: the two blocks of a systematic encoding share a junction position, and the standard bookkeeping by which redundancy is accounted for in the Hamming setting fails. We show that the resulting loss is exactly one unit in each direction, that both extremes are attained, and we identify the structural condition on the two blocks under which the loss can be avoided. The same sensitivity to adjacency reshapes the question of existence described above. The distance graph of a code is no longer determined by its Hamming distance distribution, since coordinate permutations preserve the latter but not pair distances, so the connectivity criterion that rules out strictly stronger protection in the Hamming metric does not transfer directly. We therefore study the symbol-pair analogue of the $\alpha$-distance graph and introduce two invariants of a code, the generation profile and the disconnection threshold, which record how the graph fragments as the distance parameter grows, and thereby determine which pairs of protection parameters a given linear code can realise.


 \subsection{Contribution}
    \begin{enumerate}
    \item We introduce \emph{function-correcting symbol-pair codes with data protection} (FCSPC-DP), the first framework unifying the data-protection and symbol-pair branches of the function-correcting code literature, and show how codes of this type pass between the two metrics in either direction: every Hamming-metric code of this type induces a symbol-pair code with strictly stronger protection, while every symbol-pair code yields a Hamming-metric code at roughly half the protection.

    \item We adapt the two-step construction method of the Hamming metric to the symbol-pair setting and derive the associated redundancy bounds, together with the requirement matrices adapted to the fact that the available separation comes from the codewords rather than the messages.

    \item We introduce the \emph{pair-separation constant} of a function, measuring how far apart messages carrying the same function value lie, and show that whenever it is large enough, data protection comes at no additional cost as the optimal redundancy is exactly that of the corresponding code with no data protection at all.

    \item Working with the symbol-pair analogue of the $\alpha$-distance graph, we introduce two invariants of the code, namely the \emph{generation profile} and the \emph{disconnection threshold}, which are new in both the symbol-pair and the Hamming settings.

    \item We characterise the complete set of strict parameters realisable by a given linear code, tracing the Pareto frontier between the strength of the function protection and the number of function values that can be protected. Existing constructions in the Hamming metric correspond to single points of this frontier.


    \item We give explicit constructions of function-correcting symbol-pair codes with data protection for functions whose values vary little on small neighbourhoods of the message space. The symbol-pair weight function is treated separately, where its arithmetic structure yields a stronger bound than the general result.

    \item We extend the classical Plotkin and sphere-packing bounds to this setting. The Plotkin-type bound depends on the function only through the sizes of its level sets. The sphere-packing bound is developed in three stages, in which the data-protection requirement is shown to strengthen the bound, and the final form is explicit enough to be evaluated directly.
    \end{enumerate}

\subsection{Organisation}

    Section~\ref{sec:prelim} recalls the symbol-pair metric, function-correcting codes in the Hamming and symbol-pair metric, and the
    Cayley graph and its properties. Section~\ref{sec:fcspc-dp} introduces function-correcting symbol-pair codes with data protection, establishes the relationship between the two metrics for the function minimum pair-distance, and develops the joint pair-distance requirement matrices together with the resulting bounds on optimal redundancy. Section~\ref{sec:construction} presents the two-step construction and the associated coded requirement matrices, introduces the pair-separation constant, and identifies the functions for which data protection requires no additional redundancy. Section~\ref{sec:invariants} studies the symbol-pair $\alpha$-distance graph and introduces the generation profile and the disconnection threshold, relating them to their Hamming counterparts. Section~\ref{sec:classes} gives explicit constructions for pair-locally bounded functions and for the symbol-pair weight function. Section~\ref{sec:bounds} extends the Plotkin and sphere-packing bounds to this setting, and Section~\ref{sec:conclusion} concludes the paper.

 \subsection{Notation}
    Throughout this paper, $q$ is a prime power and $\mathbb{F}_q$ the finite field with $q$ elements. We use $\mathbb{N}$ and $\mathbb{N}_0$ to denote the set of natural numbers and non-negative integers, respectively. For any matrix $\boldsymbol{D}$, we use $[\boldsymbol{D}]_{ij}$ to indicate the $(i,j)th$ entry of $\boldsymbol{D}$. For a positive integer $n$ we write $[n]=\{1,2,\dots,n\}$, and all coordinate indices are taken modulo the relevant length. For $x=(x_0,\dots,x_{n-1})\in\mathbb{F}_q^{n}$, $w_H(x)$ and $w_p(x)$ denotes its Hamming weight and symbol-pair weight respectively, and $d_H(x,y)$ and $d_p(x,y)$ denote the Hamming distance and symbol-pair distance respectively. A block code $C\subseteq\mathbb{F}_q^{n}$ of size $M$ with minimum pair-distance $d_p(C)$ is an $(n,M,d_p)_q$ code and a linear code of dimension $k$ is an $[n,k,d_p]_q$ code.
%
%
%
%

\section{Preliminaries}\label{sec:prelim}
    
 \subsection{Symbol-Pair Metric}  
    The symbol-pair distance between two vectors of $\mathbb{F}_q^n$ can be defined in terms of the Hamming distance of the symbol-pair vectors as follows:
    \begin{definition}[Symbol-pair distance and symbol-pair weight~\cite{cassuto2011codes}]\label{def:sym-pair}
     For a vector ${x} = (x_0, x_1, \ldots, x_{n-1}) \in \mathbb{F}_q^n$, we define its \emph{symbol pair vector} $\pi({x})$ as the vector of length $n$ over the alphabet $\mathbb{F}_q^2$ given by
    \[
    \pi({x}) = \left( (x_0, x_1), (x_1, x_2), \ldots, (x_{n-2}, x_{n-1}), (x_{n-1}, x_0) \right),
    \]
    where the last coordinate introduces cyclic wrapping.
    The \emph{symbol pair distance} $d_p({x}, {y})$ between two vectors ${x}, {y} \in \mathbb{F}_q^n$ is defined as the Hamming distance between their symbol pair vectors:
    \[
    d_p({x}, {y}) = \left| \left\{ i \in \{0, \ldots, n-1\} : (x_i, x_{i+1}) \neq (y_i, y_{i+1}) \right\} \right|,
    \]
    with indices taken modulo $n$, and the symbol-pair weight of a vector ${y} \in \mathbb{F}_q^n$ is defined as the Hamming weight of its symbol-pair vector:
    \[
        w_p(y) = w_H(\pi(y)) = (\left| \left\{ i \in \{0, \ldots, n-1\} : (y_i, y_{i+1}) \neq (0,0) \right\} \right|,
    \]
    where the indices are again taken modulo $n$.
    \end{definition}



    \begin{definition}[Symbol-pair code]
    \label{def:pair_code}
    A \emph{symbol pair code} $C$ of length $n \geq 2$ over $\mathbb{F}_q$ is a nonempty subset of $\mathbb{F}_q^n$. The set
    \[
    \pi(C) = \left\{ \pi(\mathbf{x}) : \mathbf{x} \in C \right\} \subseteq (\mathbb{F}_q^2)^n,
    \]
    is called the \emph{symbol pair representation} of $C$.
    \end{definition}

    \begin{definition}[Minimum symbol-pair distance]
    \label{def:min_pair_dist}
    For a symbol pair code $C \subseteq \mathbb{F}_q^n$, the \emph{minimum symbol pair distance}, denoted $d_p(C)$, is defined as
    \[
    d_p(C) = \min_{\substack{\mathbf{x}, \mathbf{y} \in C \\ \mathbf{x} \neq \mathbf{y}}} d_p(\mathbf{x}, \mathbf{y}).
    \]
    \end{definition}

    \begin{remark}
    For a linear symbol pair code $C$ (i.e., a linear subspace of $\mathbb{F}_q^n$), the minimum distance simplifies to the minimum weight of any nonzero codeword:
    \[
    d_p(C) = \min_{\mathbf{x} \in C \setminus \{\mathbf{0}\}} w_p(\mathbf{x}),
    \]
    \end{remark}
    The following lemma states the error-detection and error-correction capabilities of symbol-pair codes.
    \begin{lemma}
    \label{lem:pair_correction}
    Let $C \subseteq \mathbb{F}_q^n$ be a code of length $n$ with minimum symbol pair distance $d_p(C) = d$. Then:
    \begin{enumerate}
    \item $C$ can detect any symbol pair error of weight at most $d-1$.
    \item $C$ can correct any symbol pair error of weight at most $\left\lfloor \frac{d-1}{2} \right\rfloor$.
    \end{enumerate}
    \end{lemma}

    The next lemma establishes a relation between the symbol-pair distance of a concatenation and the symbol-pair distances of its two blocks; it is the technical device that lets one separate a message from its redundancy in the pair metric.

    \begin{lemma}[\cite{xia2024function}]
    \label{lem:symbol-pair-ineq}
        Let $\boldsymbol{u} = (\boldsymbol{u}^{(1)}, \boldsymbol{u}^{(2)}) \in \mathbb{F}_q^{m+r}$ and $\boldsymbol{v} = (\boldsymbol{v}^{(1)}, \boldsymbol{v}^{(2)}) \in \mathbb{F}_q^{m+r}$, where $\boldsymbol{u}^{(1)} = (x_0, \dots, x_{m-1})$, $\boldsymbol{u}^{(2)} = (x_m, \dots, x_{m+r-1})$, $\boldsymbol{v}^{(1)} = (v_0, \dots, v_{m-1})$ and $\boldsymbol{v}^{(2)} = (v_m, \dots, v_{m+r-1})$. Then

        $d_p(\boldsymbol{u}^{(1)}, \boldsymbol{v}^{(1)}) + d_p(\boldsymbol{u}^{(2)}, \boldsymbol{v}^{(2)}) - 1 \leq d_p(\boldsymbol{u}, \boldsymbol{v}) \leq d_p(\boldsymbol{u}^{(1)}, \boldsymbol{v}^{(1)}) + d_p(\boldsymbol{u}^{(2)}, \boldsymbol{v}^{(2)}) + 1$.
    \end{lemma}

    \begin{lemma}\label{lem:monotone}
    Let $u_1, v_1 \in \mathbb{F}_q^m$ and $u_2, v_2 \in \mathbb{F}_q^l$ with $m,l \geq 1$. Then, 
    \[
        d_p((u_1, u_2), (v_1, v_2)) \geq d_p(u_1, v_1).
    \]
        
    \end{lemma}

\subsection{Function-Correcting Symbol Pair Codes}
    We now recall the function-correcting symbol pair codes and the associated pair-distance-requirement matrices.
    \begin{definition}[Function-correcting symbol-pair codes (FCSPCs)~\cite{xia2024function}]
    \label{def:FCSPC}
    An encoding function $\mathrm{Enc} : \mathbb{F}_q^k \to \mathbb{F}_q^{k+r}$ with
    \[
    \mathrm{Enc}(x) = (x, p(x)), \quad x \in \mathbb{F}_q^k
    \]
    defines a function-correcting code for the function $f : \mathbb{F}_q^k \to \mathrm{Im}(f)$ if for all $x_1, x_2 \in \mathbb{F}_q^k$ with $f(x_1) \neq f(x_2)$, it holds that
    \[
    d_p(\mathrm{Enc}(x_1), \mathrm{Enc}(x_2)) \geq 2t + 1.
    \]    
    \end{definition}

    \begin{definition}[Pair-distance requirement matrices(PDRM)~\cite{xia2024function}]
    \label{def:PDRM}
        Consider $M$ vectors $x_0, \ldots , \boldsymbol{x_{M-1}} \in {\mathbb{F}^k_q}$. Then, 
        $\boldsymbol{D}_f^{(1)}(t, \boldsymbol{u}_1, \dots, \boldsymbol{u}_M)$ and $\boldsymbol{D}_f^{(2)}(t, \boldsymbol{u}_1, \dots, \boldsymbol{u}_M)$ are $M \times M$ matrices with entries

        $[\boldsymbol{D}_f^{(1)}(t, \boldsymbol{u}_1, \dots, \boldsymbol{u}_M)]_{ij} = \begin{cases} [2t - d_p(\boldsymbol{u}_i, \boldsymbol{u}_j)]^+, & \text{if } f(\boldsymbol{u}_i) \neq f(\boldsymbol{u}_j), \\ 0, & \text{otherwise} \end{cases}$

        and

        $[\boldsymbol{D}_f^{(2)}(t, \boldsymbol{u}_1, \dots, \boldsymbol{u}_M)]_{ij} = \begin{cases} [2t + 2 - d_p(\boldsymbol{u}_i, \boldsymbol{u}_j)]^+, & \text{if } f(\boldsymbol{u}_i) \neq f(\boldsymbol{u}_j), \\ 0, & \text{otherwise.} \end{cases}$
    \end{definition}

    \begin{definition}[Irregular-pair-distance code]\label{def:Dp-code}
        Let $M \in \mathbb{N}$ and let $D \in \mathbb{N}_0^{M \times M}$ be a matrix with non-negative integer entries. A code $\mathcal{P} = \{\bm{p}_1, \bm{p}_2, \dots, \bm{p}_M\}$ of cardinality $M$ is said to be a \emph{$D$-irregular-pair-distance code}, or a \emph{$D_p$-code} for short, if there exists an ordering of its codewords such that
        \[
        d_p(\bm{p}_i, \bm{p}_j) \;\geq\; [D]_{ij} \qquad \text{for all } i, j \in [M],
        \]
        where $[D]_{ij}$ is the $(i,j)$-th entry of $D$.
    \end{definition}

    For $D \in \mathbb{N}_0^{M \times M}$, let $N_p(D)$ denote the smallest integer $r$ for which a $D_p$-code of length $r$ exists. For the case when $[D]_{ij} = D$ for all $i \neq j$, we denote this quantity by $N_p(M, D)$.

   \subsection{Function-Correcting Codes} 
    \begin{definition}[Function-correcting codes (FCCs)~\cite{10132545}]
    \label{def:FCC}
    An encoding function $\mathrm{Enc} : \mathbb{F}_q^k \to \mathbb{F}_q^{k+r}$ with
    \[
    \mathrm{Enc}(x) = (x, p(x)), \quad x \in \mathbb{F}_q^k
    \]
    defines a function-correcting code for the function $f : \mathbb{F}_q^k \to \mathrm{Im}(f)$ if for all $x_1, x_2 \in \mathbb{F}_q^k$ with $f(x_1) \neq f(x_2)$, it holds that
    \[
    d(\mathrm{Enc}(x_1), \mathrm{Enc}(x_2)) \geq 2t + 1.
    \]    
    \end{definition}
    \begin{definition}[Function-Correcting Codes with Data Protection (FCCs-DP)~\cite{rajput2025function}]
    \label{def:FCC-DP}
        A systematic encoding $C_f: \mathbb{F}_q^k \rightarrow \mathbb{F}_q^{k+r}$ defines an $(f: d_d, d_f)_H$-function-correcting code with data protection for a function $f:\mathbb{F}_q^k \rightarrow \Img(f)$  if $d_H(C_f) \geq d_d$ and $d_H^f(C_f) \geq d_f.$
        where $d_d$ and $d_f$ are two non-negative integers such that $d_d \leq d_f$.
    \end{definition}

    \begin{definition}[Optimal Redundancy~\cite{rajput2025function}]
     \label{def: opt_r_H}
     Let $f : \mathbb{F}_q^k \to \mathrm{Im}(f)$ be a function, and let $d_d, d_f \in \mathbb{N}$ with $d_d \leq d_f$. The optimal redundancy of an $(f : d_d, d_f)_H$-FCC, denoted by $r_H^f(k : d_d, d_f)$, is defined as
     \[
        r_H^f(k, \, d_d, \, d_f) = \min\{ r\;| \; \exists \; C_f: \mathbb{F}_q^k \rightarrow\mathbb{F}_q^{k+r} \text{ such that } d_H^f(C_f) \geq d_f \text{ and } d_H(C_f) \geq d_d\}.
     \]
    \end{definition}

    \subsection{Cayley Graphs}

    This section reviews the necessary background on Cayley graphs and their relevant properties. These will serve as key tools in our subsequent analysis, enabling the characterization of the distance graph of linear symbol-pair codes, the identification of strict FCSPC-DP, and the derivation of bounds on their optimal redundancy. Additionally, we revisit the Cayley graph characterization of the $\alpha$-distance graph for linear codes under the Hamming metric, which will provide a basis for comparison with our symbol-pair results.

    \begin{definition}
    \label{def:cayley}
        Let $G$ be a group with identity element $e$ and let $S$ be a non-empty subset of $G \setminus \{e\}$ that is closed under inversion, i.e., $S^{-1} = S$. The Cayley graph of group $G$ with respect to the set $S$, denoted by $\operatorname{Cay}(G, S)$, is the undirected graph whose vertex set is $G$, and whose two distinct vertices $g$ and $h$ are adjacent iff $gh^{-1} \in S$.
    \end{definition}

    One important property of the Cayley graph that we will use in later sections is as follows.

    \begin{proposition}\label{prop:cayley-connectivity}
        Let $G$ be a group and $S$ be a non-empty subset of $G$ which is closed under inversion and does not contain the identity. If $S$ does not generate $G$, then the graph $\operatorname{Cay}(G, S)$ is disconnected, and its connected components are precisely the cosets of the subgroup generated by $S$. 
    \end{proposition}

    \begin{corollary}\label{cor:cayley-components}
    Under the hypotheses of Proposition~\ref{prop:cayley-connectivity}, and assuming additionally that \(G\) is finite, the following hold:
    \begin{enumerate}
    \item All connected components of \(\operatorname{Cay}(G, S)\) have equal size, namely \(|H| = |\langle S \rangle|\).
    \item The number of connected components is given by
    \[
    k(\operatorname{Cay}(G, S)) = \frac{|G|}{|\langle S \rangle|}.
    \]
    \end{enumerate}
    \end{corollary}

    \begin{lemma}[\cite{rajput2025function}]\label{lem:cayley-hamming}
    Let \(C\) be a linear $[n,k,d_p(C)]_q$ code and let \(\alpha \ge d_H(C)\) be a positive integer. Define
    \[
    S_\alpha^H := \{ c \in C : 0 < w_H(c) \le \alpha \}.
    \]
    Then 
    \[
        G_H^{\alpha}(C) \cong \operatorname{Cay}(C, S_\alpha^H).
    \]
    \end{lemma}


\section{Function-Correcting Symbol-Pair Codes with Data Protection(FCSPC-DP)}\label{sec:fcspc-dp}
    Before introducing FCSPCs with data protection, we formally define two distance notions. The first is the classical symbol-pair distance of a code, restated here in the systematic FCSPC setting for completeness. The second is the distance, which measures how well a code separates codewords belonging to different function-value classes.
 
    \begin{definition}[Minimum symbol-pair distance]
        \label{def:min-dist-data}
        Let $f : \mathbb{F}_q^k \to \mathrm{Im}(f)$ be a function and let $C_f : \mathbb{F}_q^k \to \mathbb{F}_q^{k+r}$ be a systematic encoding. The {minimum symbol-pair distance} of $C_f$, denoted $d_p(C_f)$, is defined as
        \[
            d_p(C_f)
            \;=\;
            \min_{\substack{x_1,\, x_2 \,\in\, \mathbb{F}_q^k \\ x_1 \neq x_2}}
            d_p\!\left(C_f(x_1),\, C_f(x_2)\right).
        \]
    \end{definition}
 
    \noindent
    The minimum distance $d_p(C_f)$ governs how many symbol-pair errors can be corrected in the data, regardless of the function $f$.  In particular, a code with $d_p(C_f) \geq d_d$ corrects up to $t_d = \lfloor(d_d-1)/2\rfloor$ errors in any transmitted codeword.
    

    \begin{definition}[Function minimum pair-distance]
    \label{def:min-dist-function}
    Let $f : \mathbb{F}_q^k \to \mathrm{Im}(f)$ be a function and let $C_f : \mathbb{F}_q^k \to \mathbb{F}_q^{k+r}$ be a systematic encoding. The {function minimum pair-distance} of $C_f$, denoted $d_p^f(C_f)$, is defined as
    \[
        d_p^f(C_f)
    \;=\;
    \min_{\substack{x_1,\, x_2 \,\in\, \mathbb{F}_q^k \\
                    f(x_1) \neq f(x_2)}}
    d_p\!\left(C_f(x_1),\, C_f(x_2)\right).
    \]
    \end{definition}
    \noindent
    In other words, $d_p^f(C_f)$ is the minimum symbol-pair distance taken {only} over pairs of codewords whose messages map to {different} function values.  Pairs within the same function-value class contribute nothing to this minimum. If $d_p^f(C_f) \geq d_f$, then $f(x)$ is uniquely determined by any channel output within symbol-pair distance $t_f = \lfloor (d_f-1)/2 \rfloor$ of the transmitted codeword.
    
    We record the elementary relationship between the two quantities.

    \begin{remark}
    \label{rem:dd-leq-df}
    For any function $f : \mathbb{F}_q^k \to \mathrm{Im}(f)$ with $|\mathrm{Im}(f)| \geq 2$ and any systematic encoding $C_f : \mathbb{F}_q^k \to \mathbb{F}_q^{k+r}$, we have
    \[
    d_p(C_f) \;\leq\; d_p^f(C_f).
    \]
    This follows immediately from the set inclusion
    \[
    B
    \;=\;
    \bigl\{(x_1,x_2)\in\mathbb{F}_q^k\times\mathbb{F}_q^k
    \mid f(x_1)\neq f(x_2)\bigr\}
    \;\subseteq\;
    \mathcal{A}
    \;=\;
    \bigl\{(x_1,x_2)\in\mathbb{F}_q^k\times\mathbb{F}_q^k
    \mid x_1\neq x_2\bigr\},
    \]

    Consequently, the constraint $d_d \leq d_f$ imposed in Definition~\ref{def:FCSPC-DP} is a consequence of the fact that no code can have a smaller function minimum pair-distance than its classical minimum distance.
    \end{remark}
    
    \noindent
    Cassuto and Blaum~\cite{cassuto2011codes} established a fundamental relationship between the Hamming and symbol-pair metrics for the classical minimum distance of a code, recalled in the following lemma. We show that this relationship extends naturally to the function minimum pair-distance (Definition~\ref{def:min-dist-function}), thereby providing a direct bridge between Hamming-FCCs (Definition~\ref{def:FCC}) and their symbol-pair counterparts.

    \begin{lemma}[\cite{cassuto2011codes}]
    \label{lem:hamming-pair-distance}
    Let $x, y \in \mathbb{F}_q^n$ with $x \neq y$.
    \begin{enumerate}
    \item If $0 < d_H(x,y) < n$, then
    $1 + d_H(x,y) \;\leq\; d_p(x,y) \;\leq\; 2\, d_H(x,y)$.
    \item If $d_H(x,y) = n$, then $d_p(x,y) = n$.
    \end{enumerate}
    Consequently, for a code $C \subseteq \mathbb{F}_q^n$ with $0 < d_H(C) < n$ we have
    $1 + d_H(C) \leq d_p(C) \leq 2\, d_H(C)$.
    \end{lemma}
    
    The following proposition shows that the same relation holds for the function minimum pair-distance.
    

    \begin{proposition}
    \label{prop:pair-hamming-func-dist-reln}
    Let $f : \mathbb{F}_q^k \to \mathrm{Im}(f)$ be non-constant and let $C_f : \mathbb{F}_q^k \to \mathbb{F}_q^{n}$ be a systematic encoding, $n = k+r$. Denote by $d_H^f(C_f)$ and $d_p^f(C_f)$ the function minimum distances of $C_f$ with respect to the Hamming and symbol-pair metrics, respectively, and assume $0 < d_H^f(C_f) < n$. Then
    \[
    1 + d_H^f(C_f) \;\leq\; d_p^f(C_f) \;\leq\; 2\, d_H^f(C_f).
    \]
    \end{proposition}

    \begin{proof}
    Let $c_x, c_y \in C_f$ such that, $d_p^f(C_f) = d_p(c_x, c_y)$ and $f(x) \neq f(y)$. Since $0 < d_H(C) <~n$ and $c_x \neq c_y$, hence by Lemma~\ref{lem:hamming-pair-distance}, $d_p(c_x, c_y) \geq 1 + d_H(c_x, c_y)$. Also since $f(x) \neq f(y)$, $d_H(c_x, c_y) \geq d_H^f(C)$, and so
    \begin{equation}
    \label{eq:lower-ham-pair}
        d_p^f(C_f) \;=\; d_p(c_x, c_y) \;\geq\; 1 + d_H(c_x, c_y) \;\geq\; 1 + d_H^f(C).
    \end{equation}
    
    \noindent
    For the upper bound, let $c_x', c_y' \in C_f$ such that, $d_H^f(C_f) = d_H(c_x', c_y')$ and $f(x) \neq f(y)$ . Again by  Lemma ~\ref{lem:hamming-pair-distance}, we have $d_p(c_x', c_y')\leq 2\, d_H(c_x', c_y') = 2\, d_H^f(C)$. Since function values corresponding to $c_x'$ and $ c_y'$ are distinct even in the symbol-pair setting, we have
    \begin{equation}
    \label{eq:upper-ham-pair}
        d_p^f(C) \;\leq\; d_p(c_x', c_y') \;\leq\; 2\, d_H^f(C).
    \end{equation}
    By combining Equations \ref{eq:lower-ham-pair} and \ref{eq:upper-ham-pair}, we get the required inequality.
    \end{proof}

    We now extend function-correcting symbol-pair codes, originally introduced for function protection only~\cite{xia2024function}, so as to incorporate data protection as well, following the Hamming-metric framework of~\cite{rajput2025function}.

    \begin{definition}[Function-correcting symbol-pair codes with data protection(FCSPCs-DP)]
    \label{def:FCSPC-DP}
        A systematic encoding $C_f: \mathbb{F}_q^k \rightarrow \mathbb{F}_q^{k+r}$ defines a $(f: d_d, d_f)_p$-function-correcting symbol-pair code with data protection for a function $f:\mathbb{F}_q^k \rightarrow \Img(f)$  if $d_p(C_f) \geq d_d$ and $d_p^f(C_f) \geq d_f.$
        where $d_d$ and $d_f$ are two positive integers such that $d_d \leq d_f$.
    \end{definition}

    \begin{definition}[Optimal Redundancy]
     \label{def: opt_r}
     Let $f : \mathbb{F}_q^k \to \mathrm{Im}(f)$ be a function, and let $d_d, d_f \in \mathbb{N}$ with $d_d \leq d_f$. The optimal redundancy of an $(f : d_d, d_f)_p$-FCSPC, denoted by $r_p^f(k ,d_d, d_f)$, is defined as
     \[
        r_p^f(k, \, d_d, \, d_f) = \min\{ r\;| \; \exists \; C_f: \mathbb{F}_q^k \rightarrow\mathbb{F}_q^{k+r} \text{ such that } d_p^f(C_f) \geq d_f \text{ and } d_p(C_f) \geq d_d\}.
     \]
    \end{definition}

    \begin{remark}
    \label{rem:r_f-r_d_f}
        Every $(f : d_d, d_f)_p$-FCSPC is a $(f : d_f)_p$-FCSPC, but the converse is not always true, indicating that $r_p^f(k, \, d_f) \leq r_p^f(k, \, d_d, \, d_f)$.
    \end{remark}
    
    

    \noindent 
    The following example illustrates how two FCSPCs for the same function $f$ and with the same redundancy length differ in the level of data protection they provide, even though both provide the same degree of function-value error correction.

    \begin{example}
\label{ex:or-function}
Let $f : \mathbb{F}_2^2 \to \{0,1\}$ be the binary OR function,
$f(x_1,x_2) = x_1 \vee x_2$, and consider the two systematic encodings
$C_f^{(1)}$ and $C_f^{(2)}$ of redundancy $r = 2$ given in Table~\ref{tab:or-example}.

\begin{table}[t]
\centering
\caption{Two encodings for the binary OR function with $r = 2$.}
\label{tab:or-example}
\renewcommand{\arraystretch}{1.3}
\begin{tabular}{cccc}
\toprule
$x$ & $f(x)$ & $C^{(1)}_f(x)$ & $C^{(2)}_f(x)$ \\
\midrule
$00$ & $0$ & $0000$ & $0000$ \\
$01$ & $1$ & $0111$ & $0111$ \\
$10$ & $1$ & $1011$ & $1011$ \\
$11$ & $1$ & $1111$ & $1101$ \\
\bottomrule
\end{tabular}
\end{table}

A direct computation gives
$d_p^f(C_f^{(1)}) = 4$, $d_p(C_f^{(1)}) = 2$, and
$d_p^f(C_f^{(2)}) = 4$, $d_p(C_f^{(2)}) = 3$.
Hence $C_f^{(1)}$ corrects one symbol-pair error at the function level but none at the
data level, whereas $C_f^{(2)}$ corrects one symbol-pair error at both levels, even
though the two encodings have the same redundancy. In particular, $C_f^{(1)}$ is a
$(f : 4)_p$-FCSPC that is not an $(f : 3, 4)_p$-FCSPC-DP.
\end{example}
 
        

    The next proposition allows one to transport any construction of $(f: d_d, d_f)_H$-FCC to $(f: d_d+1, d_f+1)_p$-FCSPC.
    \begin{proposition}
    \label{prop:hamming_symbol_red_up}
    Let $f : \mathbb{F}_q^k \to \mathrm{Im}(f)$ be a function, and let $C_f$ be an $(f : d_d, d_f)_H$-FCC of length $n = k+r$ with $d_f < n$. Then $C_f$ is an $(f : d_d+1,\, d_f+1)_p$-FCSPC. Consequently,
    \[
        r_p^f(k,\, d_d+1,\, d_f+1) \;\leq\; r_H^f(k,\, d_d,\, d_f).
    \]
    \end{proposition}
    
    \begin{proof}
        Let $x_1, x_2 \in \mathbb{F}_q^k$ with $x_1 \neq x_2$, and let $(x_1,p_1), \ (x_2,p_2) \in C_f$. Then
        
        \noindent \textbf{Case 1a: } $f(x_1) = f(x_2)$ and $d_H((x_1,p_1), (x_2,p_2)) < n$. Then from Lemma \ref{lem:hamming-pair-distance} and Definition~\ref{def:FCC-DP}
        \[
        d_p((x_1,p_1), (x_2,p_2)) \geq d_H((x_1,p_1), (x_2,p_2))+ 1 \geq d_d+1.
        \]

        \noindent \textbf{Case 1b: }
         $f(x_1) = f(x_2)$ and $d_H((x_1,p_1), (x_2,p_2)) = n$. Then $$d_p((x_1,p_1), (x_2,p_2)) = d_H((x_1,p_1), (x_2,p_2)) = n \geq d_f + 1 \geq d_d +1.$$

        \noindent \textbf{Case 2a: }
        If $f(x_1) \neq f(x_2)$ and $d_H((x_1,p_1), (x_2,p_2)) < n$. Then from Lemma \ref{lem:hamming-pair-distance} and Definition~\ref{def:FCC-DP}
        \[
        d_p((x_1,p_1), (x_2,p_2)) \geq d_H((x_1,p_1), (x_2,p_2))+ 1 \geq d_f+1.
        \]

        \noindent \textbf{Case 2b: }
        If $f(x_1) \neq f(x_2)$ and $d_H((x_1,p_1), (x_2,p_2)) = n$. Then $$d_p((x_1,p_1), (x_2,p_2)) = d_H((x_1,p_1), (x_2,p_2)) = n \geq d_f +1,$$ using $d_f < n$.

        Hence, $C_f$ is an $(f : d_d+1,\, d_f+1)_p$-FCSPC. 
        For the redundancy bound, let $C_f^*$ be an $(f:d_d,d_f)_H$-FCC attaining the optimal redundancy $r_H^f(k,d_d,d_f)$, with length $n^* = k + r_H^f(k,d_d,d_f)$ and $d_f < n^*$. By the argument above, $C_f^*$ is also a valid $(f:d_d+1,d_f+1)_p$-FCSPC with the same redundancy. Since $r_p^f(k,d_d+1,d_f+1)$ is the minimum redundancy over all valid $(f:d_d+1,d_f+1)_p$-FCSPC encodings,
        \[
        r_p^f(k,\,d_d+1,\,d_f+1) \;\leq\; r_H^f(k,\,d_d,\,d_f). 
        \]
    \end{proof}

    \noindent
    The next proposition establishes the converse relationship between the optimal redundancy of the two metrics, namely, it provides a lower bound on $r_p^f$ in terms of $r_H^f$.
    
    \begin{proposition}
    \label{lem: hamming_symbol_red_low}
    Let $f : \mathbb{F}_q^k \to \mathrm{Im}(f)$ be a function, and let $C_f$ be an $(f : d_d, d_f)_p$-FCSPC of length $n = k+r$. Then $C_f$ is an $(f : \lceil d_d/2 \rceil, \lceil d_f/2 \rceil)_H$-FCC. Consequently,
    \[
    r_p^f(k,\, d_d,\, d_f) \;\geq\; r_H^f\!\left(k,\, \left\lceil \frac{d_d}{2}\right\rceil, \left\lceil\frac{d_f}{2}\right\rceil\right).
    \]
    \end{proposition}

    \begin{proof}
    Let $x_1 \neq x_2 \in \mathbb{F}_q^k$, write $(x_1,p_1)=C_f(x_1)$, $(x_2,p_2)=C_f(x_2)$, and let $D \in \{d_d, d_f\}$ be whichever threshold applies: $D = d_d$ always, and additionally $D = d_f$ whenever $f(x_1)\neq f(x_2)$, so that $d_p((x_1,p_1),(x_2,p_2)) \geq D$ by Definition~\ref{def:FCC-DP}. Since $C_f$ is systematic, $x_1 \neq x_2$ forces $(x_1,p_1)\neq(x_2,p_2)$, so $d_H((x_1,p_1),(x_2,p_2)) \geq 1$, ruling out $d_H=0$.

    \textbf{Case A: $d_H((x_1,p_1),(x_2,p_2)) < n$.} Then $d_H \notin \{0,n\}$ and Lemma~\ref{lem:hamming-pair-distance} gives
    \[
    D \;\leq\; d_p((x_1,p_1),(x_2,p_2)) \;\leq\; 2\,d_H((x_1,p_1),(x_2,p_2)),
    \]
    so $d_H((x_1,p_1),(x_2,p_2)) \geq D/2$, and since $d_H$ is an integer, $d_H((x_1,p_1),(x_2,p_2)) \geq \lceil D/2 \rceil$.

    \textbf{Case B: $d_H((x_1,p_1),(x_2,p_2)) = n$.} By the boundary clause of Lemma~\ref{lem:hamming-pair-distance}, $d_p((x_1,p_1),(x_2,p_2)) = n$. Combined with $d_p((x_1,p_1),(x_2,p_2))\geq D$, this gives $n \geq D$, and since $\lceil D/2\rceil \leq D$ for all $D\geq0$,
    \[
    d_H((x_1,p_1),(x_2,p_2)) = n \geq D \geq \left\lceil D/2 \right\rceil.
    \]

    In both cases $d_H((x_1,p_1),(x_2,p_2)) \geq \lceil D/2\rceil$. Taking $D=d_d$ gives
$d_H \geq \lceil d_d/2\rceil$ for all $x_1\neq x_2$; taking $D=d_f$ (valid whenever
$f(x_1)\neq f(x_2)$) gives $d_H \geq \lceil d_f/2\rceil$ for cross-class pairs. Hence
$C_f$ is an $(f:\lceil d_d/2\rceil, \lceil d_f/2\rceil)_H$-FCC-DP.

For the redundancy bound, let $C_f^*$ be an optimal $(f:d_d,d_f)_p$-FCSPC, with
redundancy $r_p^f(k,d_d,d_f)$. By the above, $C_f^*$ is also a valid
$(f:\lceil d_d/2\rceil,\lceil d_f/2\rceil)_H$-FCC-DP with the same redundancy. Since
$r_H^f(k,\lceil d_d/2\rceil,\lceil d_f/2\rceil)$ is the minimum redundancy over all
such encodings,
\[
r_H^f\!\left(k,\left\lceil\tfrac{d_d}{2}\right\rceil,\left\lceil\tfrac{d_f}{2}\right\rceil\right) \;\leq\; r_p^f(k,d_d,d_f). 
\]
\end{proof}

    Next example establishes the bounds of Proposition~\ref{lem: hamming_symbol_red_low} and Proposition~\ref{prop:hamming_symbol_red_up} using the examples of FCC with data protection from the paper~\cite{rajput2025function}.
    \begin{example}
        Consider a function $f: \mathbb{F}_2^3 \to \{0,1,2,3\}$, where $f(x)$ is the position of the least frequent bit in the binary vector $x \in \mathbb{F}_2^3$, i.e., $f(000) = f(111) = 0$, \ $f(100)= f(011) = 1$, \ $f(010) = f(101) = 2$, and \ $f(001) = f(110) = 3$. For $d_d = 3$ and $d_f = 5$, \cite[Example 5]{rajput2025function} shows that the optimal redundancy of function $f$ with respect to the Hamming metric is $6$ and the code $C_f =\{000000000, 111000000, 100111100, 011111100, 010110011, 101110011, 001001111, 110001111\}$ achieves this optimal redundancy and satisfies both the minimum distance requirements for both data and function protection. Therefore from Proposition~\ref{prop:hamming_symbol_red_up} we can say that for $d_d = 4$ and $d_f = 6$, the optimal redundancy of $f$ in the symbol-pair metric is bounded above by $6$, i.e., $r_p^f(3; \ 4,6) \leq 6$ and the code $C_f$ is an $(f; \ 4,6)_p$-FCSPC as one can easily verify that $d_p(C_f) = 4$ and $d_p^f(C_f) = 7>6$.
    \end{example}

    \begin{definition}[Joint pair-distance matrices(J-PDM)] 
    \label{def:J-PDM}
    Consider $M$ vectors $x_0, \ldots , {x_{M-1}} \in {\mathbb{F}^k_q}$. Then,  $\boldsymbol{D}_f^{(1)}(d_d,d_f, x_0, \ldots , x_{M-1})$ and
    $\boldsymbol{D}_f^{(2)}(d_d,d_f, x_0, \ldots , x_{M-1})$  are $M \times M$  symbol-pair distance matrices with entries
 
    \begin{align*}
        &[\boldsymbol{D}_f^{(1)}(d_d,d_f, x_0, \ldots , x_{M-1})]_{ij}=   
        \begin{cases}
           [d_d - 1 - d_p(x_i, x_j)]^+ , & if \; x_i \neq x_j \text{ and } \; f (x_i) = f(x_j),\\
           [d_f -1  - d_p(x_i, x_j)]^+ , & if \quad f (x_i) \neq f(x_j),\\
           0, & otherwise.
        \end{cases}\\  
    & and \\
        &[\boldsymbol{D}_f^{(2)}(d_d,d_f, x_0, \ldots, x_{M-1})]_{ij}=   
        \begin{cases}
           [d_d + 1 - d_p(x_i, x_j)]^+ , & if \; x_i \neq x_j \text{ and } \; f (x_i) = f(x_j),\\
           [d_f + 1 - d_p(x_i, x_j)]^+ , & if \quad f (x_i) \neq f(x_j),\\
           0, & otherwise.
        \end{cases}
    \end{align*}
    \end{definition}

 \subsection{Bounds on Optimal Redundancy via the Joint Pair-Distance Matrices}
    
    \begin{theorem}
    \label{thm:jpdm-sandwich}
        For any function $f : \mathbb{F}_q^k \to \text{Im}(f)$ and $\{x_1, \dots, x_{q^k}\}$, we have 
        \[
        N_p(\boldsymbol{D}_f^{(1)}(d_d, d_f, x_1, \dots, x_{q^k})) \leq r^f_p(k, d_d, d_f) \leq N_p(\boldsymbol{D}_f^{(2)}(d_d,d_f, x_1, \dots, x_{q^k})).
        \]
    \end{theorem}
    \begin{proof}
        If $f$ is constant, all three quantities equal zero, and the inequalities are trivially satisfied.  Assume that $f$ is not constant.  Write $D^{(s)} = \boldsymbol{D}_f^{(s)}(d_d, d_f, x_1, \ldots, x_{q^k})$ for $s \in \{1, 2\}$ for ease of notation.
 
        \noindent
        Suppose for a contradiction that $r = r_f^p(k,d_d,d_f) < N_p(D^{(1)})$. Let $\mathrm{Enc} : \mathbb{F}_q^k \to \mathbb{F}_q^{k+r}$,
        $x_i \mapsto (x_i, p_i)$, be an optimal $(f : d_d, d_f)_p$-FCSPC with redundancy $r$. Since $r < N_p(D^{(1)})$, the set $\{p_1, \ldots, p_{q^k}\}$ is not a $D^{(1)}_p$-code, so there exist indices $i \neq j$ such that $d_p(p_i, p_j) < [D^{(1)}]_{ij}$.
        Consider the two cases arising from Definition~\ref{def:J-PDM}.
 
        \noindent
        \textbf{Case 1:} Let $x_i \neq x_j$ and $f(x_i) = f(x_j)$. Then $[D^{(1)}]_{ij} = [d_d - 1 - d_p(x_i, x_j)]^+ > d_p(p_i, p_j) >0$, so and $d_p(p_i, p_j) < d_d - 1 - d_p(x_i, x_j)$.
        By Lemma~7 of \cite{xia2024function}
        \[
            d_p\!\bigl(\mathrm{Enc}(x_i),\,\mathrm{Enc}(x_j)\bigr)
            \;\leq\;
            d_p(x_i,x_j) + d_p(p_i,p_j) + 1
            \;<\;
            d_p(x_i,x_j) + (d_d - 1 - d_p(x_i,x_j)) + 1
            = d_d,
        \]
        contradicting the data-protection requirement $d_p(\mathrm{Enc}(x_i), \mathrm{Enc}(x_j)) \geq d_d$.

        \noindent
        \textbf{Case 2:} Let $f(x_i) \neq f(x_j)$. Then $[D^{(1)}]_{ij} = [d_f - 1 - d_p(x_i,x_j)]^+ > d_p(x_i,x_j) > 0$, so $d_p(p_i,p_j) < d_f - 1 - d_p(x_i,x_j)$. Again by Lemma~7 of \cite{xia2024function},
        \[
        d_p\!\bigl(\mathrm{Enc}(x_i),\,\mathrm{Enc}(x_j)\bigr)
        \;\leq\;
        d_p(x_i,x_j) + d_p(p_i,p_j) + 1
        \;<\;
        d_p(x_i,x_j) + (d_f - 1 - d_p(x_i,x_j)) + 1
        = d_f,
        \]
        contradicting the function-protection requirement $d_p(\mathrm{Enc}(x_i), \mathrm{Enc}(x_j)) \geq d_f$.
 
        Hence, in both cases we reach a contradiction, so $r_f^p(k,d_d,d_f) \geq N_p(D^{(1)})$.
 
        \noindent
        We now show that $r_f^p(k,d_d,d_f) \leq N_p(D^{(2)})$.
        Let $P = \{p_1, \ldots, p_{q^k}\} \subseteq \mathbb{F}_q^r$ be a $D^{(2)}_p$-code of length $r = N_p(D^{(2)})$, and define the
        encoding function $\mathrm{Enc} : x_i \mapsto (x_i, p_i)$. We verify both distance conditions for an $(f : d_d, d_f)_p$-FCSPC.\\
        \noindent
        \textit{Data protection (within-class pairs).} Let $x_i \neq x_j$ with $f(x_i) = f(x_j)$. By Lemma~\ref{lem:symbol-pair-ineq}
        \begin{equation}
        \label{eq:lb-enc}
        d_p\!\bigl(\mathrm{Enc}(x_i),\,\mathrm{Enc}(x_j)\bigr)
        \;\geq\;
        d_p(x_i,x_j) + d_p(p_i,p_j) - 1.
        \end{equation}
        \begin{enumerate}
         \item If $d_p(x_i, x_j) \geq d_d + 1$: then $[D^{(2)}]_{ij} = 0$
         and $d_p(p_i, p_j) \geq 0$, so by \eqref{eq:lb-enc},
        $d_p(\mathrm{Enc}(x_i), \mathrm{Enc}(x_j)) \geq d_p(x_i,x_j) - 1 \geq d_d$.
 
        \item If $d_p(x_i, x_j) \leq d_d$: then
        $[D^{(2)}]_{ij} = d_d + 1 - d_p(x_i, x_j) > 0$ and
        $d_p(p_i, p_j) \geq d_d + 1 - d_p(x_i, x_j)$, so by \eqref{eq:lb-enc},
        $d_p(\mathrm{Enc}(x_i), \mathrm{Enc}(x_j))
        \geq d_p(x_i, x_j) + d_d + 1 - d_p(x_i, x_j) - 1 = d_d$.
        \end{enumerate}
        In both sub-cases, $d_p(\mathrm{Enc}(x_i), \mathrm{Enc}(x_j)) \geq d_d$.
 
        \noindent
       \textit{Function protection (cross-class pairs).} Let $f(x_i) \neq f(x_j)$. By \eqref{eq:lb-enc},
        \begin{enumerate}
        \item If $d_p(x_i, x_j) \geq d_f + 1$: then $[D^{(2)}]_{ij} = 0$
        and $d_p(\mathrm{Enc}(x_i), \mathrm{Enc}(x_j)) \geq d_p(x_i,x_j) - 1 \geq d_f$.
 
        \item If $d_p(x_i, x_j) \leq d_f$: then $[D^{(2)}]_{ij} = d_f + 1 - d_p(x_i, x_j)$ and $d_p(p_i, p_j) \geq d_f + 1 - d_p(x_i, x_j)$, so $d_p(\mathrm{Enc}(x_i), \mathrm{Enc}(x_j)) \geq d_p(x_i, x_j) + d_f + 1 - d_p(x_i, x_j) - 1 = d_f$.
        \end{enumerate}
        In both sub-cases, $d_p(\mathrm{Enc}(x_i), \mathrm{Enc}(x_j)) \geq d_f$.
 
        Hence $\mathrm{Enc}$ defines a valid $(f : d_d, d_f)_p$-FCSPC with redundancy $N_p(D^{(2)})$, giving $r_f^p(k, d_d, d_f) \leq N_p(D^{(2)})$.
    \end{proof}

     \begin{corollary}
         For any function $f: \mathbb{F}_q^k \rightarrow \Img{f}$ and $\{x_1, x_2, \ldots x_m\} \subseteq \mathbb{F}_q^k$, we have 
         \[
            N_p(\boldsymbol{D}_f^{(1)}(d_d, d_f, x_1, \dots, x_{m})) \leq r_p^f(k, d_d, d_f).
         \]
         and $ r_p^f(k,d_d,d_f) \geq d_f - 3 $ for $|\Img{f}| \geq 2$.   
     \end{corollary}

    \begin{proof}
        It can be verified easily that redundancy vectors of optimal $(f, d_d, d_f)_p$-FCSPC is a $\boldsymbol{D}_f^{(1)}(d_d, d_f, x_1, \dots, x_{m})$-code. Hence,
        \begin{equation}
        \label{eq:lower_bound_optimal_red}
            N_p(\boldsymbol{D}_f^{(1)}(d_d, d_f, x_1, \dots, x_{m})) \leq r_f^p(k, d_d, d_f).
        \end{equation}
        For any non-constant function, there always exist two vectors, say $x$ and $y$, such that $d_p(x,y) = 2$ and $f(x) \neq f(y)$. Therefore, from the equation~\ref{eq:lower_bound_optimal_red} we get,
        \[
            r_f(k,t_d,t_f) \geq N_p(\boldsymbol{D}_f^{(1)}(d_d, d_f, x,y)) = N_p(2, d_f - 3) = d_f - 3.
        \]
        The last equality follows from the fact that a repetition code of length $d_f - 3$ has a minimum pair distance $d_f - 3$.
    \end{proof}

\section{Explicit Construction of Function-Correcting Symbol-Pair Codes with Data-Protection}\label{sec:construction}

    \noindent
    The following construction method extends the two-step approach for constructing FCC (\cite{rajput2025function}) to the symbol-pair metric, incorporating data protection.
    
    \begin{construction}
    \label{construction-2-step}
    In order to construct an FCPSC for a function $f: \mathbb{F}_q^k \to \mathrm{Im}(f)$, that provides protection against up to $t_d = \left\lfloor \frac{d_d-1}{2} \right\rfloor$ symbol-pair errors in the data and up to $t_f = \left\lfloor \frac{d_f-1}{2} \right\rfloor$ symbol-pair errors in the function values, the following two-fold construction mechanism can be adopted.
    
    \begin{itemize}
    \item \textbf{Step 1:} Select an $[n,k]_q$ linear symbol-pair error-correcting code $C$, with minimum pair-distance $d_p(C) \geq d_d$. Let $G$ be a generator matrix of size $k \times n$. For any $x \in \mathbb{F}_q^k$, the corresponding codeword is given by $c_x = xG$.

    \item \textbf{Step 2:} Construct a FCSPC based on the set $\{c_x \mid x \in \mathbb{F}_q^k\}$, rather than directly on the message vectors $x$. That is, define a systematic encoding $C'_f : C \to \mathbb{F}_q^{n+r'}$ such that for any $x_1, x_2 \in \mathbb{F}_q^k$ with $f(x_1) \neq f(x_2)$, the following condition holds:
    \[
        d_p\!\left( C'_f(c_{x_1}),\, C'_f(c_{x_2}) \right) \;\geq\; d_f.
    \]
    \end{itemize}
    
    Then, the resulting mapping $C_f : \mathbb{F}_q^k \to \mathbb{F}_q^{n+r'}$ defined by $C_f(x) = C'_f(c_x)$ is an $(f : d_d, d_f)_p$-FCSPC with total redundancy $r_s = n - k + r'$.

    It is straightforward to verify that the encoding $C_f$ described above satisfies the properties of an $(f: d_d, d_f)_p$-FCSPC. Since $C$ is an $[n,k,d_d]_q$ linear symbol-pair code, every distinct pair $x_1, x_2 \in \mathbb{F}_q^k$ satisfies $d_p(c_{x_1}, c_{x_2}) \geq d_d$. 
    
    Two cases now follow:
    \begin{itemize}
    \item {Same-class regime} ($f(x_i) = f(x_j)$): Let $x_i \neq x_j \in \mathbb{F}_q^k$. Then,
    \[
        d_p\!\left( C_f(x_1),\, C_f(x_2) \right) \;\geq\;d_p(c_{x_1}, c_{x_2})  \;\geq\; d_d,
    \]

    \item {Cross-class regime} ($f(x_i) \neq f(x_j)$). By the Step~2 design,
    \[
        d_p\!\left( C_f(x_1),\, C_f(x_2) \right) \;=\; d_p\!\left( C'_f(c_{x_1}),\, C'_f(c_{x_2}) \right) \;\geq\; d_f \geq d_d,
    \]
    \end{itemize}
    Combining both cases, $C_f$ is an $(f: d_d, d_f)_p$-FCSPC of length $n + r'$ and redundancy $r_s = n - k + r'$.
    \end{construction}
        
              
        
              
        
        

    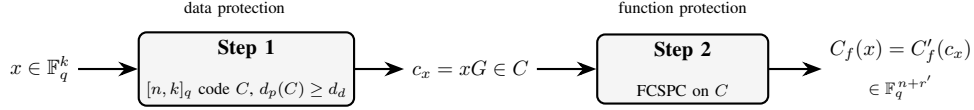
\begin{figure}[t]
    \centering
    \begin{tikzpicture}[scale=0.78,transform shape,
      blk/.style={draw,thick,rounded corners=3pt,minimum height=1.05cm,
                  minimum width=2.9cm,align=center,font=\small},
      ar/.style={-{Stealth[length=2.6mm]},thick},
      font=\small
    ]
    \node[font=\small] (u) at (0,0) {$x\in\mathbb{F}_q^{k}$};
    \node[blk,fill=gray!8] (s1) at (3.5,0) {\textbf{Step 1}\\[1pt]
          \scriptsize \scriptsize $[n,k]_q$ code $C$, $d_p(C)\ge d_d$};
    \node[font=\small] (cu) at (7.3,0) {$c_x=xG\in C$};
    \node[blk,fill=gray!8] (s2) at (10.9,0) {\textbf{Step 2}\\[1pt]
          \scriptsize FCSPC on $C$};
    \node[font=\small,align=center] (out) at (14.6,0)
          {$C_f(x)=C'_f(c_x)$\\[1pt]\scriptsize $\in\mathbb{F}_q^{\,n+r'}$};
    \draw[ar] (u)--(s1); \draw[ar] (s1)--(cu);
    \draw[ar] (cu)--(s2); \draw[ar] (s2)--(out);
    \node[font=\scriptsize,anchor=south] at (3.3,0.75) {data protection};
    \node[font=\scriptsize,anchor=south] at (10.9,0.75) {function protection};
    \end{tikzpicture}
    \caption{The two-step encoding of Construction~\ref{construction-2-step}. Step~1
    maps the message through a systematic linear code of minimum pair-distance at least
    $d_d$, fixing the baseline separation for the data; Step~2 applies a systematic
    FCSPC to the resulting codeword rather than to the message, raising the separation
    of cross-class pairs to $d_f$.}
    \label{fig:two-step}
\end{figure}

    In the original framework, the pair-distance requirement matrix is indexed by the message vectors, whose mutual pair-distances already contribute to the required separation. In Construction~\ref{construction-2-step} that role is played instead by the codewords $c_x=xG$, so the residual requirement must be measured against $d_p(c_x,c_y)$ rather than $d_p(x,y)$. This motivates the following coded version of pair-distance requirement matrices.


    

        \begin{definition}[Coded pair-distance requirement matrices (CPDRM)]
        \label{def:CPDRM}
        Let $C$ be a linear $[n,k]_q$ code with generator matrix $G$, write $c_x = xG$, and let
        $x_1, \ldots, x_M \in \mathbb{F}_q^{k}$ be distinct. The \emph{coded pair-distance
        requirement matrices}
        $\boldsymbol{D}^{(1)}_{C,f}(d_f;\, x_1, \ldots, x_M)$ and
        $\boldsymbol{D}^{(2)}_{C,f}(d_f;\, x_1, \ldots, x_M)$ are the $M \times M$ matrices with
        entries
        \begin{align*}
        \bigl[\boldsymbol{D}^{(1)}_{C,f}(d_f;\, x_1, \ldots, x_M)\bigr]_{ij}
        &=
        \begin{cases}
            \bigl[\,d_f - 1 - d_p(c_{x_i}, c_{x_j})\,\bigr]^{+},
                & \text{if } f(x_i) \neq f(x_j),\\[2pt]
            0, & \text{otherwise},
        \end{cases}
        \\[6pt]
        \bigl[\boldsymbol{D}^{(2)}_{C,f}(d_f;\, x_1, \ldots, x_M)\bigr]_{ij}
        &=
        \begin{cases}
            \bigl[\,d_f + 1 - d_p(c_{x_i}, c_{x_j})\,\bigr]^{+},
                & \text{if } f(x_i) \neq f(x_j),\\[2pt]
            0, & \text{otherwise},
        \end{cases}
        \end{align*}
        \end{definition}
        
        \begin{definition}
        \label{def:coded-func-pair-dist}
        For a function $f : \mathbb{F}_q^k \to \text{Im}(f)$, the coded function pair-distance between $f_i, f_j \in \text{Im}(f)$ is defined as:

        $$d_p^{C}(f_i, f_j) = \min_{x_1, x_2 \in \mathbb{F}_q^k} \{d_{p}(c_{x_1}, c_{x_2}) \mid f(x_1) = f_1, f(x_2) = f_2\}.$$
        \end{definition}



        \begin{definition}[Coded function pair-distance matrices]
        \label{def:coded-func-matrices}
        Let $\mathrm{Im}(f) = \{f_1, \ldots, f_E\}$ with $E = |\mathrm{Im}(f)|$. The
        \emph{coded function pair-distance matrices}
        $\boldsymbol{E}^{(1)}_{C,f}(d_f)$ and $\boldsymbol{E}^{(2)}_{C,f}(d_f)$ are the
        $E \times E$ matrices with entries
        \begin{align*}
        \bigl[\boldsymbol{E}^{(1)}_{C,f}(d_f)\bigr]_{ij}
        &=
        \begin{cases}
            \bigl[\,d_f - 1 - d^{C}_{p}(f_i, f_j)\,\bigr]^{+}, & \text{if } i \neq j,\\[2pt]
            0, & \text{if } i = j,
        \end{cases}
        \\[6pt]
        \bigl[\boldsymbol{E}^{(2)}_{C,f}(d_f)\bigr]_{ij}
        &=
        \begin{cases}
            \bigl[\,d_f + 1 - d^{C}_{p}(f_i, f_j)\,\bigr]^{+}, & \text{if } i \neq j,\\[2pt]
            0, & \text{if } i = j.
        \end{cases}
        \end{align*}
        \end{definition}

        \begin{proposition}
    \label{cor:two-step-redundancy}
    Let $f:\mathbb{F}_q^{k}\to\mathrm{Im}(f)$ and let $d_d\le d_f$ be positive integers. Let $C$ be a linear $[n,k,d_d]_q$ symbol-pair code with generator matrix $G$ and $c_x=xG$, and fix an ordering $x_1,\dots,x_{q^k}$ of $\mathbb{F}_q^{k}$. Then
    \[
    r^f_p(k,d_d,d_f)\;\le\;(n-k)\;+\; N_p\Bigl(\boldsymbol{D}^{(2)}_{C,f} \bigl(d_f;\,x_1,\dots,x_{q^k}\bigr)\Bigr).
    \]
    \end{proposition}

    \begin{proof}
    Write $r''=N_p\bigl(\boldsymbol{D}^{(2)}_{C,f}(d_f;\,x_1,\dots, x_{q^k})\bigr)$. By definition of $N_p(\cdot)$ there exist vectors $p_{1},\dots,p_{q^k}\in\mathbb{F}_q^{\,r''}$ with
    \[
        d_p(p_{i},p_{j})\;\ge\; \Bigl[\boldsymbol{D}^{(2)}_{C,f} \bigl(d_f;\,x_1,\dots,x_{q^k}\bigr)\Bigr]_{i,j} \;=\;\max\bigl\{\,d_f+1 -d_p(c_{x_i},c_{x_j}),\;0\,\bigr\}
    \]
    for all $i\ne j$ with $f(x_i)\ne f(x_j)$. Apply Construction~\ref{construction-2-step} with $C$ as the first step and $p_{1},\dots,p_{q^k}$ as the second, that is, set
    \[
    C_f(x_i)\;=\;\bigl(c_{x_i},\,p_{i}\bigr)\;\in\;\mathbb{F}_q^{\,n+r''}.
    \]
    We verify the two requirements of FCSPC-DP (Definition~\ref{def:FCSPC-DP}).

    \emph{Data protection.} For $i\ne j$ we have $d_p(c_{x_i},c_{x_j})\ge d_p(C)=d_d$, so by  (Lemma~\ref{lem:monotone})
    \[
        d_p\bigl(C_f(x_i),C_f(x_j)\bigr)\;\ge\;d_p(c_{x_i},c_{x_j})\;\ge\;d_d .
    \]

    \emph{Function protection.} Let $f(x_i)\ne f(x_j)$. If $d_p(c_{x_i},c_{x_j})\ge d_f+1$ the matrix entry vanishes and we have  $d_p\bigl(C_f(x_i),C_f(x_j)\bigr)\ge d_p(c_{x_i},c_{x_j}) \ge\;d_f+1>d_f$. Otherwise, the entry equals $d_f+1-d_p(c_{x_i},c_{x_j})$, and the Lemma~\ref{lem:symbol-pair-ineq} yields
    \[
        d_p\bigl(C_f(x_i),C_f(x_j)\bigr) \;\ge\;d_p(c_{x_i},c_{x_j})+d_p(p_{i},p_{j})-1 \;\ge\;d_p(c_{x_i},c_{x_j})+\bigl(d_f+1-d_p(c_{x_i},c_{x_j})\bigr)-1 \;=\;d_f .
    \]

Hence $C_f$ is an $(f:d_d,d_f)_p$-FCSPC-DP of length $n+r''$ and redundancy
$(n+r'')-k=(n-k)+r''$. Since $r^f_p(k,d_d,d_f)$ is the minimum redundancy over
all $(f:d_d,d_f)_p$-FCSPC-DPs, the bound follows.
\end{proof}
\subsection{Block-aware Analysis of Function-Correcting Symbol-Pair Codes with Data Protection}

    Let $f : \mathbb{F}_q^k \to \mathrm{Im}(f)$ be a function. Then the level sets of $f$ form a natural partition of the message space.

    \begin{definition}[Pair-separation constant]
    \label{def:pair-sep-constant}
    Let $f:\mathbb{F}_q^{k}\to\mathrm{Im}(f)$ be non-injective, and let
    \[
        P_f\;=\;\bigl\{\,\{x,y\}\;:\;x,y\in\mathbb{F}_q^{k},\ x\neq y,\ f(x)=f(y)\,\bigr\}
    \]
    denote the non-empty set of same-class pairs. The \emph{pair-separation constant} of $f$ is
    \[
        \delta_{p}(f)\;=\;\min_{\{x,y\}\in P_f} d_p(x,y),
    \]
    the least pair-distance between two distinct messages carrying the same
    function value. For an integer $\delta$, we say that $f$ is \emph{$d$-pair-separated} if $\delta_{p}(f)\ge d$.
    \end{definition}


    The following example considers a function that is a $d_d$-pair separated function.
    \begin{example}
    Let $f:\mathbb{F}_2^{k}\to\mathrm{Im}(f)$, $k\ge3$, be the Hamming weight function $f(x)=w_H(x)$. Two distinct vectors of equal Hamming weight differ in at least two coordinates, so $d_H(x,y) \ge 2$. For the pair metric, let $x \ne y$ with $w_H(x)=w_H(y)$. If $d_H(x,y)<k$ then Lemma~\ref{lem:hamming-pair-distance} gives $d_p(x,y)\ge d_H(x,y)+1\ge3$. If $d_H(x,y)=k$ then $y=\bar x$, whence $d_p(x,y)=k\ge3$. Thus $\delta_{p}(f)\ge3$. Hence the Hamming weight function is $3$-pair-separated.
    \end{example}
    
    In the two-step construction method~(\ref{construction-2-step}), the ECC pays an extra $n - k$ redundancy uniformly across all $q^k$ messages, irrespective of whether $f$ already enforces the intra-class pair-separation. For $d_d$-pair-separated functions, the intra-class constraint is met by the geometry of $f$, and this overhead is paid unnecessarily.\\

    \noindent

    \noindent
    We now show that for such $f$, the optimal redundancy in the pair-metric data-protection problem coincides with that of the plain pair-metric FCSPC.

    \begin{proposition}
    \label{prop:dd-sep-equality}
    Let $f$ satisfy $\delta_{p}(f)\ge d_d$ and let $d_f\ge d_d$. Then
    \[
    r^f_p(k,d_d,d_f)\;=\;r^f_p(k,d_f).
    \]
    \end{proposition}
 
    \begin{proof}
    By Remark~\ref{rem:r_f-r_d_f} it suffices to prove $r^f_p(k,d_d,d_f) \le r^f_p(k,d_f)$.

    Let $r' = r^f_p(k, d_f)$ and let $C'_f : \mathbb{F}_q^{k} \to \mathbb{F}_q^{\,k+r'}$ be a systematic $(f : d_f)_p$-FCSPC of optimal redundancy. We verify that $C'_f$ is in fact an $(f : d_d, d_f)_p$-FCSPC-DP.

    \emph{Cross-class pairs.} If $f(x_1) \neq f(x_2)$, then $d_p\bigl(C'_f(x_1), C'_f(x_2)\bigr) \ge d_f \ge d_d$, so both requirements hold.

    \emph{Same-class pairs.} If $x_1 \neq x_2$ and $f(x_1) = f(x_2)$, then $C'_f$ is systematic, so Lemma~\ref{lem:monotone} gives
    \[
    d_p\bigl(C'_f(x_1), C'_f(x_2)\bigr) \;\ge\; d_p(x_1,x_2) \;\ge\; \delta_p(f) \;\ge\; d_d .
    \]

    Hence $d_p(C'_f) \ge d_d$ and $d_p^f(C'_f) \ge d_f$, so $C'_f$ is a valid $(f : d_d,d_f)_p$-FCSPC-DP of redundancy $r^f_p(k,d_f)$, and minimality of $r^f_p(k,d_d,d_f)$ gives the claim.
    \end{proof}

    \begin{remark}
    For an $f$ satisfying the hypothesis of Proposition~\ref{prop:dd-sep-equality}, the $n-k$ redundancy of Step~1 Construction~\ref{construction-2-step} is pure overhead: data protection is already implied by the function-protection requirement together with the geometry of $f$, and the construction should not be used.
    \end{remark}

    \begin{proposition}
    \label{prop:pair-jdrm-equals-pdrm}
    Let $f$ satisfy $\delta_p(f) \ge d_d + 1$. Then for any ordering $x_1,\dots,x_{q^k}$ of $\mathbb{F}_q^{k}$ and each $s \in \{1,2\}$,
    \[
    \boldsymbol{D}^{(s)}_{f}(d_d, d_f;\, x_1,\dots,x_{q^k})
    \;=\;
    \boldsymbol{D}^{(s)}_{f}(d_f;\, x_1,\dots,x_{q^k}),
    \]
    where the matrices on the right are the pair-distance requirement matrices of Definition~\ref{def:PDRM}. Consequently the bounds of Theorem~\ref{thm:jpdm-sandwich} for $r^f_p(k,d_d,d_f)$ coincide with the corresponding bounds for $r^f_p(k,d_f)$.
    \end{proposition}

    \begin{proof}
    Let $x_i \neq x_j$.

    \emph{Same-class pairs.} Suppose $f(x_i) = f(x_j)$. Since $\delta_p(f) \ge d_d + 1$, we have $d_p(x_i,x_j) \ge d_d + 1$, so
    \[
    \bigl[d_d - 1 - d_p(x_i,x_j)\bigr]^{+} = 0
    \quad\text{and}\quad
    \bigl[d_d + 1 - d_p(x_i,x_j)\bigr]^{+} = 0 ,
    \]
    which are the entries of $\boldsymbol{D}^{(1)}_{f}(d_d,d_f;\cdot)$ and $\boldsymbol{D}^{(2)}_{f}(d_d,d_f;\cdot)$ respectively. The corresponding PDRM entries are $0$ by definition.

    \emph{Cross-class pairs.} If $f(x_i) \neq f(x_j)$, the J-PDM entries are $[d_f - 1 - d_p(x_i,x_j)]^{+}$ and $[d_f + 1 - d_p(x_i,x_j)]^{+}$, which are precisely the PDRM entries with parameter $d_f$.

    Diagonal entries vanish in all four matrices. Hence the matrices coincide.
    \end{proof}

\subsection{Bounds from the Coded Distance-Requirement Matrix}

    We now bound $N_p\bigl(\boldsymbol{D}^{(2)}_{C,f}(d_f;\cdot)\bigr)$, the quantity appearing in Proposition~\ref{cor:two-step-redundancy}, first by a scalar symbol-pair quantity and then by a Hamming quantity, for which bounds and tables are available. We begin by relating the coded matrix back to the joint matrix of Definition~\ref{def:J-PDM}.

  \begin{theorem}\label{thm:pair-matrix-two-step}
    Let $C$ be a systematic $[n,k,d_d]_q$ symbol-pair code, given by $c_x=(x,z_x)$ with $z_x\in\mathbb{F}_q^{\,n-k}$, and let $f:\mathbb{F}_q^{k}\to\mathrm{Im}(f)$. Then, for any ordering $x_1,\dots,x_{q^k}$ of $\mathbb{F}_q^{k}$,
    \[
      N_p\bigl(\boldsymbol{D}^{(1)}_{f}(d_d,d_f;\,x_1,\dots,x_{q^k})\bigr)
      \;\le\;
      N_p\bigl(\boldsymbol{D}^{(2)}_{C,f}(d_f;\,x_1,\dots,x_{q^k})\bigr)+n-k .
    \]
    \end{theorem}

    \begin{proof}
    Let $\boldsymbol{D}^1=\boldsymbol{D}^{(1)}_{f}(d_d,d_f;\,x_1,\dots,x_{q^k})$ and $\boldsymbol{D}^2=\boldsymbol{D}^{(2)}_{C,f}(d_f;\,x_1,\dots,x_{q^k})$.
    
    Let $\mathcal{P}_2=\{p_x : x\in\mathbb{F}_q^{k}\}$ be an ordered $\boldsymbol{D}_p^2$-code of length $N_p(\boldsymbol{D}^2)$, and set
    \[
      \mathcal{P}_1=\bigl\{\,\bar p_x=(z_x,p_x)\;:\;x\in\mathbb{F}_q^{k}\,\bigr\} \subseteq\mathbb{F}_q^{\,(n-k)+N_p(\boldsymbol{D}^2)} .
    \]
    We claim that $\mathcal{P}_1$ is a $\boldsymbol{D}^1$-code, i.e.\  $d_p(\bar p_x,\bar p_y)\ge[\boldsymbol{D}^1]_{x,y}$ for all $x, y \in \mathbb{F}_q^k$ and the rows and columns of $\boldsymbol{D}^1$ indexed by the vectors in $\mathbb{F}_q^k$.


    Let $x\neq y$ with $f(x)=f(y)$. By Lemma~\ref{lem:monotone} ,
    \[
    d_p(\bar p_x,\bar p_y)\;\ge\;d_p(z_x,z_y) \;\ge\;d_p(c_x,c_y)-d_p(x,y)-1 \;\ge\;d_d-1-d_p(x,y)\;=\;[\boldsymbol{D}^1]_{x,y},
    \]
    where the second inequality follows from the result of Lemma~\ref{lem:symbol-pair-ineq} and the third inequality follows using $d_p(c_x,c_y)\ge d_d$.

    Let $f(x)\neq f(y)$. Since $\mathcal{P}_2$ is a $\boldsymbol{D}^2$-code, $d_p(p_x,p_y)\ge d_f+1-d_p(c_x,c_y)$.
    \begin{align*}
        d_p(\bar p_x,\bar p_y) & \ge\;d_p(z_x,z_y)+d_p(p_x,p_y)-1\\
        & \;\ge\;d_p(z_x,z_y)+\bigl(d_f+1-d_p(c_x,c_y)\bigr)-1\\
        &\;\ge\;d_f-\bigl(d_p(c_x,c_y)-d_p(z_x,z_y)\bigr)\\
        & \;\ge\;d_f-1-d_p(x,y),
    \end{align*}

    Hence $\mathcal{P}_1$ is a $\boldsymbol{D}^1$-code of length $(n-k)+N_p(\boldsymbol{D}^2)$, and the claim follows.
    \end{proof}

    \begin{remark}
    Theorem~\ref{thm:pair-matrix-two-step} also follows in two lines by combining Proposition~\ref{cor:two-step-redundancy} with the lower bound of Theorem~\ref{thm:jpdm-sandwich}. The proof given above is constructive, exhibiting an explicit $\boldsymbol{D}^{(1)}_{f}$-code of the stated length rather than merely bounding $N_p$.
    \end{remark}


    \begin{theorem}
    \label{thm:coded-drm-bound}
    Let $f : \mathbb{F}_q^{k} \to \mathrm{Im}(f)$, let $C$ be a linear $[n,k]_q$ code with $d_p(C) \ge d_d$, let $d_f \ge d_d$, and let $x_1,\dots,x_M \in \mathbb{F}_q^{k}$ be distinct, $2 \le M \le q^k$. Then
    \[
    N_p\bigl(\boldsymbol{D}^{(2)}_{C,f}(d_f;\, x_1,\dots,x_M)\bigr)
    \;\le\; N_p\bigl(M,\; d_f - d_d + 1\bigr).
    \]
    \end{theorem}

    \begin{proof}
    For distinct $x_i, x_j$ we have $d_p(c_{x_i},c_{x_j}) \ge d_p(C) \ge d_d$, so for $i \neq j$ with $f(x_i) \neq f(x_j)$,
    \[
    \bigl[\boldsymbol{D}^{(2)}_{C,f}(d_f;\, x_1,\dots,x_M)\bigr]_{ij}
    = \max\bigl\{ d_f + 1 - d_p(c_{x_i},c_{x_j}),\, 0 \bigr\}
    \le \max\bigl\{ d_f + 1 - d_d,\, 0 \bigr\}
    = d_f - d_d + 1,
    \]
    the last equality by $d_f \ge d_d$; the remaining entries are $0$. Hence every entry is at most $d_f - d_d + 1$, so any $M$ vectors with pairwise symbol-pair distance at least $d_f - d_d + 1$ form a $\boldsymbol{D}^{(2)}_{C,f}(d_f;\,x_1,\dots,x_M)$-code, and the bound follows.
    \end{proof}

    \begin{corollary}
    \label{cor:coded-drm-hamming}
    Let $C$ be an $[n,k,d_d]_p$ symbol-pair code, let $d_f>d_d$, and suppose $M>q$. Then
    \[
    N_p\bigl(\boldsymbol{D}^{(2)}_{C,f}(d_f;\, x_1,\dots,x_M)\bigr)
    \;\le\; N_p\bigl(M,\, d_f - d_d + 1\bigr)
    \;\le\; N_H\bigl(M,\, d_f - d_d\bigr).
    \]
    \end{corollary}

    \begin{corollary}[Redundancy of the two-step construction]
    \label{cor:two-step-hamming}
    Under the hypotheses of Corollary~\ref{cor:coded-drm-hamming}, the two-step construction applied to $C$ yields
    \[
    r^f_p\bigl(k,\,2t_d+1,\,2t_f+1\bigr) \;\le\; (n-k)\;+\;N_H\bigl(M,\,2(t_f-t_d)-1\bigr).
    \]
    \end{corollary}

\section{Strictness and the $\alpha$-Pair-Distance Graph}
\label{sec:invariants}

    Recall that an $(f:d_d,d_f)_p$-FCSPC simultaneously guarantees a minimum pair-distance $d_d$ between all codewords and a minimum pair-distance $d_f$ between codewords whose messages carry distinct function values. An FCSPC-DP is called a strict FCSPC-DP when $d_f>d_d$. The boundary case $d_f=d_d$ imposes no requirement beyond that of an ordinary symbol-pair code. Beyond asking merely whether a fixed code with $d_p(C)=d_d$ admits a strict FCSPC-DP for some function, we ask the sharper quantitative question: how strong can the function protection $d_f$ be made? We show that both questions are governed by a one-parameter family of graphs associated with the code, and, for linear codes, by a single subcode-generation invariant.

 \subsection{The $\alpha$-Pair-Distance Graph}

    \begin{definition}[$\alpha$-Pair-distance graph]\label{def:alpha-pair-graph} Let $C$ be a symbol-pair code with minimum pair-distance $d_p(C)$. For an integer $\alpha\ge d_p(C)$, the $\alpha$-pair-distance graph $G_p^{\alpha}(C)$ has vertex set $C$, with distinct codewords $c_1,c_2$ adjacent if and only if $d_p(c_1,c_2)\le\alpha$.
    \end{definition}

    The next result uses the connectedness of the $\alpha$-pair-distance graph to characterize when a code serves as a strict FCSPC-DP for a desired protection level.

    \begin{theorem}\label{thm:alpha-connected-nonexistence}
    Let $C$ be an $(n,q^k,d_p(C))_q$ symbol-pair code and let $\alpha\ge d_p(C)$. If $G_p^{\alpha}(C)$ is connected, then $C$ cannot be a strict $(f:d_p(C),d_f)_p$-FCSPC for any $f:\mathbb{F}_q^k\to\mathrm{Im}(f)$ with $|\mathrm{Im}(f)|\ge 2$ and $d_f>\alpha$.
    \end{theorem}

    \begin{proof}
    Assume, for the sake of contradiction, that there exists a symbol-pair code \(C\) whose $\alpha$-pair-distance graph is connected and is a strict $(f:d_p(C),d_f)_p$-FCSPC with \(d_f > \alpha\) and \(|\operatorname{Im}(f)| \ge 2\). 

    For each element \(a \in \operatorname{Im}(f)\), define the fibre
    \[
    C_a := \{c_x \in C : f(x) = a\}.
    \]
    Since \(|\operatorname{Im}(f)| \ge 2\), there exist at least two distinct indices \(a,b \in \operatorname{Im}(f)\) such that both \(C_a\) and   \(C_b\) are non-empty. 

    Now take arbitrary codewords \(c_x \in C_a\) and \(c_y \in C_b\) with \(a \neq b\). By the definition of FCSPC the distance between codewords originating from distinct fibres is bounded below by \(d_f\). Hence
    \[
    d_p(c_x, c_y) \ge d_f.
    \]
    Given that \(d_f > \alpha\), we obtain the strict inequality
    \[
    d_p(c_x, c_y) > \alpha.
    \]

    Recall that in the graph \(G_p^{\alpha}(C)\), two distinct vertices are adjacent if and only if their pair-distance is at most \(\alpha\). Since \(d_p(c_x, c_y) > \alpha\), the pair \(\{c_x, c_y\}\) cannot form an edge in \(G_p^{\alpha}(C)\). Therefore, no edge of \(G_p^{\alpha}(C)\) connects vertices from different fibres \(C_a\) and \(C_b\) whenever \(a \neq b\).

     This implies that the vertex set \(C\) decomposes into disjoint, non-empty blocks \(C_a\) (indexed by \(a \in \operatorname{Im}(f)\)) such that there are no edges between distinct blocks. Hence, the graph \(G_p^{\alpha}(C)\) is disconnected, leading to a contradiction.

    \end{proof}

    The preceding proof establishes a necessary structural condition for a code to be a strict FCSPC-DP for a given function \(f\). Specifically, if \(C\) is such a strict FCSPC-DP code, then the \(\alpha\)-pair-distance graph \(G_p^{\alpha}(C)\) must have at least as many connected components as the cardinality of the image of \(f\). This yields the following non-existence result for strict FCSPC-DP.
    
    \begin{theorem}\label{thm:alpha-components-nonexistence}
    Let $C$ be an $(n,q^k,d_p(C))_p$ symbol-pair code and $\alpha\ge d_p(C)$. If $G_p^{\alpha}(C)$ has $Q$ connected components, then $C$ cannot be a strict $(f:d_p(C),d_f)_p$-FCSPC for any $f$ with $|\mathrm{Im}(f)|\ge Q+1$ and $d_f>\alpha$.
    \end{theorem}

    The converse direction gives an existence criterion for arbitrary codes to be strict FCSPC-DP.

    \begin{theorem}\label{thm:pair-existence} 
    Let $C$ be an $(n,q^k,d_p(C))_q$ code with $d_p(C)=d_d$, let $f:\mathbb{F}_q^k\to\mathrm{Im}(f)$ with $E=|\mathrm{Im}(f)|$, and let $d_f>d_d$. Suppose the connected components of $G_p^{\,d_f-1}(C)$ can be grouped into $E$ pairwise disjoint unions $C_1,\dots,C_E$ with $|C_i|=|f^{-1}(a_i)|$ for some ordering $a_1,\dots,a_E$ of $\mathrm{Im}(f)$. Then $C$ is a strict $(f:d_d,d_f)_p$-FCSPC.
    \end{theorem}

    \begin{proof}
    Let \(\{a_1,\dots,a_m\}\) be ordering of image set  of $f$. For each \(i\in\{1,\dots,m\}\), define a bijection  
    \[
    \varphi_i : f^{-1}(a_i) \longrightarrow C_i,
    \]
    where the sets \(C_i\) are pairwise disjoint connected component of $G_p^{\,d_f-1}(C)$. Define the encoding map  
    \[
    C_f:\mathbb{F}_q^k \longrightarrow \bigcup_{i=1}^m C_i
    \]
    by  
    \[
    C_f(x) := \varphi_i(x) \quad \text{whenever } f(x)=a_i.
    \]

    We first verify injectivity. Suppose \(x,x'\in C\) with \(C_f(x)=C_f(x')\).  Since the \(C_i\) are disjoint, there exists a unique index \(i\) such that \(C_f(x),C_f(x')\in C_i\). Hence \(f(x)=f(x')=a_i\). Then  
    \[
    \varphi_i(x)=\varphi_i(x')
    \]
    and because \(\varphi_i\) is a bijection, we conclude \(x=x'\).  Thus \(C_f\) is injective, so distinct messages map to distinct codewords.

    Consider two distinct messages \(x,x'\in \mathbb{F}_q^k\). We consider two cases.

    \medskip
    \noindent
    \textbf{Case 1:} \(f(x)=f(x')=a_i\) for some \(i\).  Then \(C_f(x)=\varphi_i(x)\) and \(C_f(x')=\varphi_i(x')\).  Since \(x\neq x'\) and \(\varphi_i\) is injective, we have \(C_f(x)\neq C_f(x')\).  By the minimum distance property of the original code,  
    \[
    d_p(C_f(x),C_f(x')) \ge d_p(C) = d_d.
    \]

    \medskip
    \noindent
    \textbf{Case 2:} \(f(x)=a_i\) and \(f(x')=a_j\) with \(i\neq j\). Then \(C_f(x)\in \varphi_i(f^{-1}(a_i))\) and \(C_f(x')\in \varphi_j(f^{-1}(a_j))\).  By the component assumption, these two symbols lie in distinct connected components  of \(G_p^{d_f-1}\). Since two vertices belonging to different components cannot be adjacent in \(G_p^{d_f-1}\). Therefore, 
    \[
    d_p(C_f(x),C_f(x')) > d_f-1.
    \]

    Combining both cases, we conclude that for all distinct \(x,x'\in C\),
    \[
    d_p(C_f(x),C_f(x')) \ge 
    \begin{cases}
        d_d, & \text{if } f(x)=f(x'),\\[4pt]
        d_f, & \text{if } f(x)\neq f(x').
    \end{cases}
    \]
    In particular, the encoding \(C_f\) defines a strict $(f: d_d, d_f)_p$ FCSPC-DP. This completes the proof.
    \end{proof}

     \subsection{Linear codes: Cayley Structure and the Generation profile}

    Throughout this subsection $C$ is a linear $[n,k,d_p(C)]_q$ symbol-pair code. The inherent algebraic structure of linear codes endows the \(\alpha\)-pair-distance graph with a particularly clean algebraic characterization. This algebraic formulation, in turn, affords a substantial simplification of the conditions under which a given function \(f\) admits a strict FCSPC-DP.

    \begin{lemma}\label{lem:cayley-pair}
    Let \(C\) be a linear $[n,k,d_p(C)]_q$ symbol-pair linear code and let \(\alpha \ge d_p(C)\) be a positive integer. Define
    \[
    S_\alpha^p := \{ c \in C : 0 < w_p(c) \le \alpha \}.
    \]
    Then 
    \[
        G_p^{\alpha}(C) \cong \operatorname{Cay}(C, S_\alpha^p).
    \]
    \end{lemma}

    \begin{proof}
    We first observe that \(S_\alpha^p\) is closed under additive inversion. For any \(z \in \mathbb{F}_q^n\), \(w_p(z) = w_p(-z)\), hence when \(z \in S_\alpha^p\) implies \(-z \in S_\alpha^p\). This ensures that \(S_\alpha^p\) is a valid connection set for the Cayley graph construction.

    It now suffices to show that the edge sets of the two graphs coincide. Let \(x, y \in C\) be distinct vertices. By the definition of the \(\alpha\)-pair-distance graph, \(x\) and \(y\) are adjacent in \(G_p^{\alpha}(C)\) if and only if $d_p(x, y) \le \alpha$.
    Since \(C\) is a linear code, $d_p(x, y) = w_p(x - y)$. Therefore, \(x\) and \(y\) are adjacent in \(G_p^{\alpha}(C)\) if and only if $0 < w_p(x - y) \le \alpha$, which is equivalent to \(x - y \in S_\alpha^p\). By the definition of the Cayley graph \(\operatorname{Cay}(C, S_\alpha^p)\), two vertices \(x, y \in C\) are adjacent precisely when \(x - y \in S_\alpha^p\). Hence adjacency in \(G_p^{\alpha}(C)\) coincides exactly with adjacency in \(\operatorname{Cay}(C, S_\alpha^p)\).

    Conversely, if \(x\) and \(y\) are adjacent in \(\operatorname{Cay}(C, S_\alpha^p)\), then \(x - y \in S_\alpha^p\), which gives \(w_p(x - y) \le \alpha\), and consequently \(d_p(x, y) \le \alpha\), establishing adjacency in \(G_p^{\alpha}(C)\). The equivalence of the adjacency relations therefore holds in both directions.

    Since both graphs share the same vertex set \(C\) and have identical edge sets, we conclude that
    \[
    G_p^{\alpha}(C) \cong \operatorname{Cay}(C, S_\alpha^p).
    \]
    \end{proof}

    Because the symbol-pair weight satisfies \(w_p(\lambda z) = w_p(z)\) for every non-zero scalar \(\lambda \in \mathbb{F}_q^*\), the set \(S^{\alpha}_p\) is closed under multiplication by elements of \(\mathbb{F}_q^*\). Consequently, the subgroup generated by \(S^{\alpha}_p\) in the additive group of \(C\) coincides with its \(\mathbb{F}_q\)-linear span, denoted \(\langle S^{\alpha}_p \rangle\). This observation, together with the Cayley graph characterization of \(G_p^{\alpha}(C)\) established in Lemma~\ref{lem:cayley-pair} and the general connectivity properties of Cayley graphs given in Proposition~\ref{prop:cayley-connectivity} and Corollary~\ref{cor:cayley-components}, gives the following result.

    \begin{corollary}\label{cor:cosets-pair}
    Let \(C\) be an $[n, k, d_p(C)]_q$ linear symbol-pair code, and let \(\alpha \ge d_p(C)\) be an integer. Define
    \[
        S^{\alpha}_p := \{ c \in C : 0 < w_p(c) \le \alpha \}
    \]
    and let \(\gamma_C^p(\alpha) := \dim_{\mathbb{F}_q} \langle S^{\alpha}_p \rangle\). Then the connected components of the \(\alpha\)-pair-distance graph \(G_p^{\alpha}(C)\) are precisely the cosets of the linear subspace \(\langle S^{\alpha}_p \rangle\) in \(C\). In particular:
    \begin{enumerate}
    \item Each connected component has cardinality \(q^{\gamma}\);
    \item The total number of connected components is \(q^{k-\gamma}\);
    \item The graph \(G_p^{\alpha}(C)\) is disconnected if and only if \(\langle S^{\alpha}_p \rangle \subsetneq C\), i.e., if and only if \(\gamma < k\).
    \end{enumerate}
    \end{corollary}

    The following example illustrates the results of Corollary~\ref{cor:cosets-pair}.

    \begin{example}
    Let \(C \subseteq \mathbb{F}_3^4\) be the linear symbol-pair code given as follows $$C := \{0000,\; 1211,\; 1100,\; 0111,\; 1022,\; 2200,\; 2011,\; 2122,\; 0222\}$$ with minimum symbol-pair distance \(d_p(C) = 3\).
    Consider \(\alpha = 3\), then the set of non-zero codewords of pair weight at most 3 is given by
    \[
    S^3_p = \{ c \in C : 0 < w_p(c) \le 3 \} = \{1100,\; 2200\},
    \]
    
    and linear subcode of $C$ generated by $S^3_p$ is $\langle S^{3}_p \rangle = \{0000, 1100, 2200\}$ with $ \gamma_C^p(3)=\dim\langle S^3_p\rangle = 1.$ Fig~\ref{fig:cayley-ternary} shows $G_p^{3}(C)$ has $3^{2-1}=3$ components of size $3^1$ same as stated in Corollary~\ref{cor:cosets-pair}. The connected components of $G_p^{3}(C)$ are:
    \[
    \{0000,1100,2200\},\quad \{0111,1211,2011\},\quad \{0222,1022,2122\},
    \] 
    \end{example}

    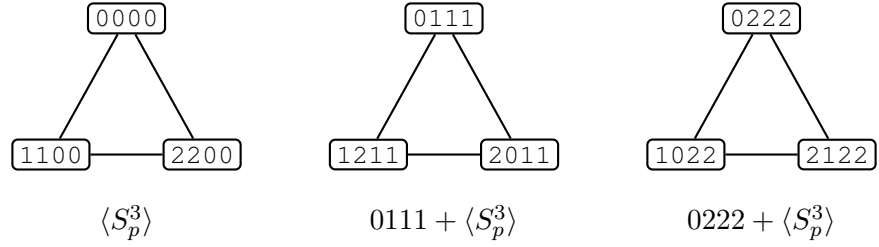
\begin{figure}[t]
    \centering
    \begin{tikzpicture}[
    scale=1.0,
    vtx/.style={draw, rounded corners=2pt, fill=white, inner sep=2.5pt,
              font=\small\ttfamily},
    every path/.style={thick}
    ]

    \begin{scope}[xshift=0cm]
  \node[vtx] (a0) at (0,1.15)      {0000};
  \node[vtx] (a1) at (-1.0,-0.65)  {1100};
  \node[vtx] (a2) at (1.0,-0.65)   {2200};
  \draw (a0)--(a1); \draw (a1)--(a2); \draw (a2)--(a0);
  \node at (0,-1.55) {$\langle S^{3}_p\rangle$};
    \end{scope}

    \begin{scope}[xshift=4.2cm]
    \node[vtx] (b0) at (0,1.15)      {0111};
    \node[vtx] (b1) at (-1.0,-0.65)  {1211};
    \node[vtx] (b2) at (1.0,-0.65)   {2011};
    \draw (b0)--(b1); \draw (b1)--(b2); \draw (b2)--(b0);
    \node at (0,-1.55) {$0111+\langle S^{3}_p\rangle$};
    \end{scope}

\begin{scope}[xshift=8.4cm]
  \node[vtx] (c0) at (0,1.15)      {0222};
  \node[vtx] (c1) at (-1.0,-0.65)  {1022};
  \node[vtx] (c2) at (1.0,-0.65)   {2122};
  \draw (c0)--(c1); \draw (c1)--(c2); \draw (c2)--(c0);
  \node at (0,-1.55) {$0222+\langle S^{3}_p\rangle$};
\end{scope}

    \end{tikzpicture}
    \caption{The Cayley graph $G^{3}_p(C)\cong\mathrm{Cay}\bigl(C,S^{3}_p\bigr)$ for the $[4,2, 3]_3$ symbol-pair code $C=\langle(1,1,0,0),(0,1,1,1)\rangle$ with $S^{3}_p=\{1100,2200\}$. The three connected components are the cosets of $\langle S^{3}_p\rangle$, each a triangle of size $|\langle S^{3}_p\rangle|=3$.}
    \label{fig:cayley-ternary}
    \end{figure}

    By Corollary~\ref{cor:cosets-pair}, the connectivity of $G_p^{\alpha}(C)$ is decided entirely by whether the codewords of pair weight at most $\alpha$ generate $C$. As $\alpha$ increases these codewords accumulate, so the subcodes they span form an increasing chain
    \[
    \langle S^p_{d_p(C)}\rangle \;\subseteq\; \langle S^p_{d_p(C)+1}\rangle
    \;\subseteq\;\cdots\;\subseteq\; \langle S^p_{n}\rangle \;=\; C ,
    \]
    which terminates at $C$ for $\alpha$ large enough. Two features of this chain carry all the information we need: the level at which it first reaches $C$, beyond which the graph is connected and no strict code exists; and the dimensions along the way, which by Corollary~\ref{cor:cosets-pair} determine the number and size of the components at each level, hence which functions are admissible. We therefore record the chain through its dimension function.
    
    \begin{definition}[Pair-generation profile and disconnection threshold]
        Let \(C\) be an $[n, k, d_p(C)]_q$ linear symbol-pair code and $\langle S^\alpha_p\rangle$ be the span of collection of all non-zero codewords of $C$ whose symbol-pair weight is at most $\alpha$ then the pair-generation profile of $C$ is the non-decreasing function defined as follows:
        \[
            \gamma_C^p(\alpha)=\dim\langle S^\alpha_p\rangle, \qquad d_p(C)\le\alpha\le n,
        \]
        and the disconnection threshold of $C$ is defined as:
        \[
            \alpha_p^{*}(C)=\max\{\alpha\ge d_p(C):\gamma_C^p(\alpha)<k\},
        \]
        with the convention $\alpha_p^{*}(C)=d_p(C)-1$ when $C$ is generated by its minimum-pair-weight codewords.
    \end{definition}

    \begin{example}
    \label{exa:pair-gen-fun-and-gen-thres}
    Consider a linear symbol-pair code $C$ having parameters $[10, 3, 3]_2$ as follows:
    \[
    \begin{array}{llll}
    0000000000\;(w_p =0) & 1100000000\;(w_p = 3) & 0001110000\;(w_p = 4) & 0000001111\;(w_p = 5)\\
    1101110000\;(w_p =7) & 1100001111\;(w_p = 7) & 0001111111\;(w_p = 8) & 1101111111\;(w_p = 10)
    \end{array}
    \]
    Then, \\
    $\langle S^3_p \rangle = \{0000000000, 1100000000\} \qquad and \quad  \gamma_C^p(3)= 1.$\\    
    $\langle S_p^4 \rangle= \{ 0000000000, 1100000000, 0001110000, 1101110000\} \qquad and \quad  \gamma_C^p(4)= 2.$\\
    $\langle S_p^5 \rangle= \{ 0000000000, 1100000000, 0001110000, 1101110000, 0000001111, 1100001111,0001111111,\\ \qquad \quad  1101111111 \} = C \qquad and \quad  \gamma_C^p(5)= 3 = k.$\\
    Thus, the disconnection threshold of $C$ is, $\alpha_p^{*}(C) = 4$.
    The components of the graph in the Figure~\ref{fig:levels-10-3} with respect to each generating set $S_\alpha^p$ verifies the above results found for each $\alpha$. 
    \end{example}

    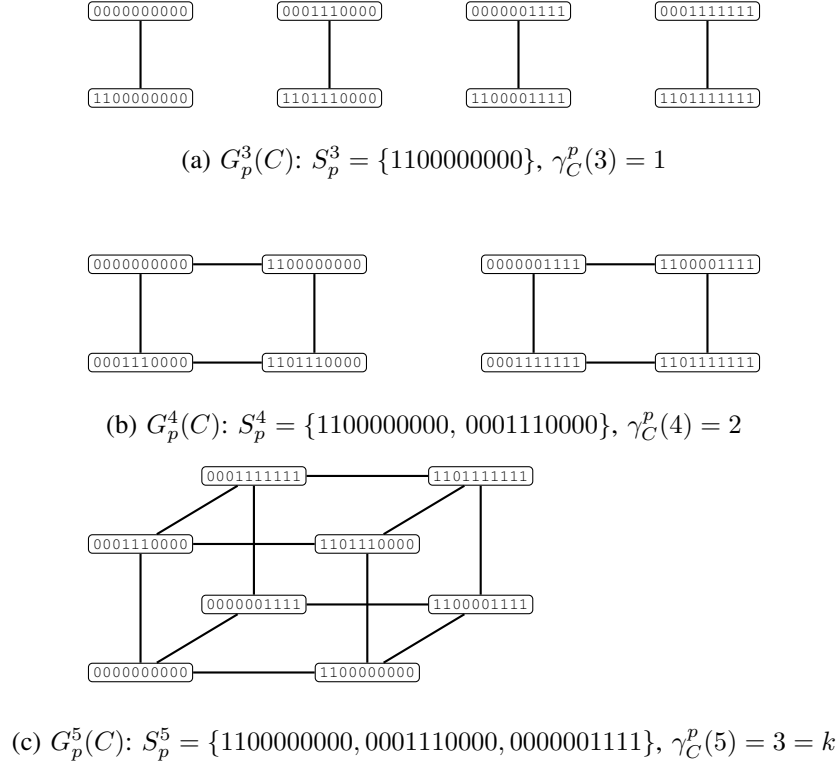
\begin{figure}[t]
    \centering
    \begin{tikzpicture}[
    gvtx/.style={draw,rounded corners=1.5pt,fill=white,inner sep=1.6pt,
            font=\tiny\ttfamily},
    gedge/.style={thick},
    gcap/.style={font=\small}
    ]
    
    \begin{scope}[yshift=0cm]
    \foreach \i/\u/\l in {0/0000000000/1100000000,
                        1/0001110000/1101110000,
                        2/0000001111/1100001111,
                        3/0001111111/1101111111}{
   \node[gvtx] (u\i) at (2.5*\i,1.15) {\u};
   \node[gvtx] (l\i) at (2.5*\i,0)    {\l};
    \draw[gedge] (u\i)--(l\i);
    }
    \node[gcap] at (3.75,-0.85){(a) $G_p^{3}(C)$: $S_p^{3}=\{1100000000\}$, $\gamma_C^p(3)=1$};
    \end{scope}

    \begin{scope}[yshift=-3.5cm]
    \node[gvtx] (p0) at (0,1.3)   {0000000000};
    \node[gvtx] (p1) at (2.3,1.3) {1100000000};
    \node[gvtx] (p2) at (2.3,0)   {1101110000};
    \node[gvtx] (p3) at (0,0)     {0001110000};
    \draw[gedge] (p0)--(p1)--(p2)--(p3)--(p0);
    
     \node[gvtx] (q0) at (5.2,1.3) {0000001111};
     \node[gvtx] (q1) at (7.5,1.3) {1100001111};
     \node[gvtx] (q2) at (7.5,0)   {1101111111};
     \node[gvtx] (q3) at (5.2,0)   {0001111111};
    \draw[gedge] (q0)--(q1)--(q2)--(q3)--(q0);
    
      \node[gcap] at (3.75,-0.85){(b) $G_p^{4}(C)$: $S_p^{4}=\{1100000000,\,0001110000\}$, $\gamma_C^p(4)=2$};
    \end{scope}
    
    \begin{scope}[yshift=-7.6cm]
     \node[gvtx] (f00) at (0,0)     {0000000000};
     \node[gvtx] (f10) at (3.0,0)   {1100000000};
     \node[gvtx] (f11) at (3.0,1.7) {1101110000};
     \node[gvtx] (f01) at (0,1.7)   {0001110000};
     \node[gvtx] (b00) at (1.5,0.9) {0000001111};
     \node[gvtx] (b10) at (4.5,0.9) {1100001111};
     \node[gvtx] (b11) at (4.5,2.6) {1101111111};
     \node[gvtx] (b01) at (1.5,2.6) {0001111111};
    
      \draw[gedge] (f00)--(f10)--(f11)--(f01)--(f00);
      \draw[gedge] (b00)--(b10)--(b11)--(b01)--(b00);
      \draw[gedge] (f00)--(b00); \draw[gedge] (f10)--(b10);
      \draw[gedge] (f11)--(b11); \draw[gedge] (f01)--(b01);
    
      \node[gcap] at (3.75,-0.95)
          {(c) $G_p^{5}(C)$: $S_p^{5}=\{1100000000,0001110000,0000001111\}$,
           $\gamma_C^p(5)=3=k$};
            \end{scope}
        \end{tikzpicture}
        \caption{The graphs $G_p^{\alpha}(C)$ for the $[10,3,3]_p$ code $C$, at $\alpha=3,4,5$. As $\alpha$ grows the generating set $S_p^{\alpha}$ gains one new independent codeword at each level, and the components merge accordingly. The graph is disconnected up to $\alpha=4$, so $\alpha_p^{*}(C)=4$.}
    \label{fig:levels-10-3}
    \end{figure}

     The preceding results on the connected components of the $\alpha$-pair-distance graph naturally lead to a complete characterization of the achievable strict protection parameters for linear symbol-pair codes. The following theorem establishes the precise range of function protection distances $d_f$ for which a given linear code can serve as a strict FCSPC-DP code.
    
    \begin{theorem}
    \label{thm:pair-tradeoff}
    Let $C$ be a linear $[n,k,d_p(C)]_q$ code with $d_d=d_p(C)$.
    \begin{enumerate}
    \item[(i)] For every integer $d_f$ with $d_d<d_f\le\alpha_p^{*}(C)+1$, the  code $C$ serves as a strict $(f:d_d,d_f)_p$-FCSPC for every function $f:\mathbb{F}_q^k\to\mathrm{Im}(f)$ satisfying
    \[
    |\mathrm{Im}(f)|\le q^{\,k-\gamma_C^p(d_f-1)} \quad\text{and}\quad q^{\,\gamma_C^p(d_f-1)}\ \big|\ |f^{-1}(a)|\ \ \forall a\in\mathrm{Im}(f).
    \]
    \item[(ii)] For $d_f>\alpha_p^{*}(C)+1$, no strict $(f:d_d,d_f)_p$-FCSPC exists for any non-constant $f$.
    \end{enumerate}
    \end{theorem}

    \begin{proof}
    We establish each part of the theorem separately.

    \medskip
    \noindent
    \textbf{(i)} Assume $d_d < d_f \le \alpha_p^*(C) + 1$. Then $d_f - 1 \le \alpha_p^*(C)$, and by the definition of the disconnection threshold, we have $\gamma_C^p(d_f - 1) < k$.

    Applying Corollary~\ref{cor:cosets-pair} with $\alpha = d_f - 1$, we conclude that the graph $G_p^{\,d_f - 1}(C)$ has exactly $q^{\,k - \gamma_C^p(d_f - 1)}$ connected components, each of which is a coset of the linear subspace  $\langle S_{d_f - 1}^p \rangle$
    where
    \[
    S_{d_f - 1}^p := \{ c \in C : 0 < w_p(c) \le d_f - 1 \}.
        \]
    Each component has cardinality $q^{\,\gamma_C^p(d_f - 1)}$.

    Now, by hypothesis, every fibre $f^{-1}(a)$ has cardinality divisible by $|H| = p^{\,\gamma_C^p(d_f - 1)}$. Since all connected components are equinumerous, we can partition the set of components into groups whose total sizes match the prescribed fibre sizes $|f^{-1}(a)|$. More precisely, because the total number of components is at least $|\operatorname{Im}(f)|$ (by the first condition) and each fibre size is a multiple of the component size, we can assign to each $a \in \operatorname{Im}(f)$ a collection of $|f^{-1}(a)| / |H|$ distinct cosets whose union is exactly $f^{-1}(a)$. This grouping ensures that each fibre is a union of connected components of $G_p^{\,d_f - 1}(C)$.

    The existence theorem, Theorem~\ref{thm:pair-existence}, then guarantees that $C$ serves as a strict $(f : d_d, d_f)_p$-FCSPC code. This completes the proof of part (i).

    \medskip
    \noindent
    \textbf{(ii)} Suppose, to the contrary, that $d_f > \alpha_p^*(C) + 1$ and that there exists a non-constant function $f$ for which $C$ is a strict $(f : d_d, d_f)_p$-FCSPC code. Then $d_f - 1 > \alpha_p^*(C)$, so by the definition of the disconnection threshold, we have
    \[
    \gamma_C^p(d_f - 1) = k.
    \]
    Consequently, $\langle S_{d_f - 1}^p \rangle = C$, and Corollary~\ref{cor:cosets-pair} implies that $G_p^{\,d_f - 1}(C)$ is connected.

    However, by Theorem~\ref{thm:alpha-connected-nonexistence} (with $\alpha = d_f - 1$), a strict FCSPC-DP code with protection distance $d_f$ cannot exist when the graph $G_p^{\,d_f - 1}(C)$ is connected, because connectedness precludes the separation of distinct fibres into different connected components. This contradicts the existence of such a non-constant function $f$. Hence, no strict code exists for $d_f > \alpha_p^*(C) + 1$.
    \end{proof}

    \begin{remark} 
        As $d_f$ increases, the function $\gamma_C^p(d_f-1)$ is non-decreasing. Consequently, the number of admissible function values, $q^{\,k-\gamma_C^p(d_f-1)},$ can only decrease, while the divisibility constraint $q^{\,\gamma_C^p(d_f-1)} \mid |f^{-1}(a)| \qquad \forall a \in \operatorname{Im}(f)$, becomes progressively more restrictive. Thus, the family
        \[
            \left(d_f,\; q^{\,k-\gamma_C^p(d_f-1)}\right) \qquad d_d < d_f \le \alpha_p^*(C) + 1,
        \]
        traces the complete Pareto frontier between the function protection distance $d_f$ and the maximum number of protectable function values classes. Each increment in $\gamma_C^p(d_f-1)$ marks a threshold at which strengthening the protection distance incurs a corresponding reduction in the number of admissible function values by a factor of $q$.
    \end{remark}

        The notions of generation profile and disconnection threshold are not specific to the symbol-pair metric. Indeed, by substituting the symbol-pair weight $w_p$ with the Hamming weight $w_H$, these definitions carry over verbatim to the Hamming metric, and the statement of Theorem~\ref{thm:pair-tradeoff} remains applicable. To facilitate comparison, we denote the generation profile and disconnection threshold in the Hamming metric by $\gamma_C^H$ and $\alpha_H^*(C)$, respectively, and in the symbol-pair metric by $\gamma_C^p$ and $\alpha_p^*(C)$. A direct comparison between these quantities is established in the following proposition.

    \begin{proposition}
    \label{prop:gamma-ham-pair}
        Let $C \subseteq \mathbb{F}_q^n$ be a linear code of dimension $k$. Then
        \begin{enumerate}
            \item[(i)] $\gamma_C^p(\alpha) \leq \gamma_C^H(\alpha-1)$ for every $\alpha \leq n-1;$
            \item[(ii)] $\gamma_C^H(\beta) \leq \gamma_C^p(2\beta)$ for every $\beta \geq 0.$
        \end{enumerate}
    \end{proposition}

    \begin{proof}
        Let $c \in S_\alpha^p$. By the definition of $S_\alpha^p$, we have $0 < w_p(c) \leq \alpha < n$, so $w_H(c) < n$. It follows that
        \[
        w_H(c) \leq w_p(c) - 1 \leq \alpha - 1.
        \]
        This implies that $S_\alpha^p \subseteq S_{\alpha-1}^H$, and therefore $\langle S_\alpha^p \rangle \subseteq \langle S_{\alpha-1}^H \rangle$.
        
        For the second result, let $c' \in S_\beta^H$. Then $0 <  w_H(c') \leq \beta$. Using the inequality $w_p(c') \leq 2w_H(c')$, we obtain
        \[
        w_p(c') \leq 2w_H(c') \leq 2\beta.
        \]
        Hence $c' \in S_{2\beta}^p$, which yields $\langle S_\beta^H \rangle \subseteq \langle S_{2\beta}^p \rangle$.
    \end{proof}

    \begin{corollary}\label{cor:threshold-sandwich}
    Let $C$ be linear with $\alpha^{*}_H(C)+1\le n-1$. Then
    \[
        \alpha^{*}_H(C)+1\;\le\;\alpha^{*}_p(C)\;\le\;2\,\alpha^{*}_H(C)+1 .
    \]
    Both bounds are attained.
    \end{corollary}
    \begin{proof}
    For the lower bound consider $\alpha=\alpha^{*}_H(C)+1$ in Proposition~\ref{prop:gamma-ham-pair}(i). Then $\gamma_C^p(\alpha^{*}_H(C)+1) \leq \gamma^{H}_C(\alpha^{*}_H(C))<k$, so $\gamma^{p}_C(\alpha)<k$, whence $\alpha^{*}_p(C)\ge \alpha^{*}_H(C)+1$. For the upper bound put $\beta=\alpha^{*}_H(C)+1$ in Proposition~\ref{prop:gamma-ham-pair}(ii), so $\gamma^{H}_C(\beta)=k$ and  $\gamma^{p}_C(2\beta)=k$, hence
    $\alpha^{*}_p(C)<2\beta=2\alpha^{*}_H(C)+2$.
    \end{proof}

    \begin{corollary}\label{cor:strictness-transfer}
Suppose the linear code $C$ serves as a strict $(f:d_d,d_f)_H$-FCC-DP with
$d_f\le n-1$. Then $C$ serves as a strict $(f:d_d,d_f+1)_p$-FCSPC, and the
number of admissible function classes does not decrease:
\[
  q^{\,k-\gamma^{p}_C(d_f)}\;\ge\;q^{\,k-\gamma^{H}_C(d_f-1)} .
\]
\end{corollary}
    \begin{proof}
    Strictness in the Hamming metric gives $\gamma^{H}_C(d_f-1)<k$. By Proposition~\ref{prop:gamma-ham-pair}(i) with $\alpha=d_f$ we obtain $\gamma^{p}_C(d_f)\le\gamma^{H}_C(d_f-1)<k$, so $G^{d_f}_p(C)$ is disconnected and Theorem~\ref{thm:pair-tradeoff} applies at function distance $d_f+1$. 
    \end{proof}

    \subsection{Consequences for optimal redundancy}

    The results of the preceding subsections test whether a given code can serve as a strict FCSPC-DP. We now show that the same structure constrains the optimal redundancy. Write $N_p(M,d)$ for the least length of a $q$-ary pair-metric code of size $M$ with minimum pair-distance at least $d$, and $N_p^{\mathrm{lin}}(q^k,d)$ for the least length of a linear such code of dimension $k$.

    At a fixed length, the graph condition is not merely a test but a characterisation, so redundancy is the minimisation of the length over all codes passing it. For linear codes, this takes a classical form.

    \begin{proposition}\label{prop:nested}
    A linear code $C\subseteq\mathbb{F}_q^n$ of dimension $k$ serves as a strict $(f:d_d,d_f)_p$-FCSPC for some non-constant $f$ if and only if there is a subcode $H\subsetneq C$ with
    \[
        d_p(C)\ \ge\ d_d,
        \qquad
    w_p(c)\ \ge\ d_f \ \ \text{for every } c\in C\setminus H .
    \]
    In that case, one may take $H=\langle S^p_{d_f-1}\rangle$, and the admissible functions are exactly those of Theorem~\ref{thm:pair-tradeoff}. Consequently
    \[
    r_{p}^f(k,d_d,d_f) \;\leq\; \min\{\,n-k \;:\; \exists\ H\subseteq C\subseteq\mathbb{F}_q^n \ \text{as above with } \dim C=k \,\}.
    \]
    \end{proposition}
    
    \begin{proof}
    If $C$ is strict then, by Corollary~\ref{cor:cosets-pair}, the classes of $f$ are unions of cosets of $H=\langle S^p_{d_f-1}\rangle$, and $H\subsetneq C$ since $G_p^{\,d_f-1}(C)$ is disconnected. Every $c\in C\setminus H$ has $w_p(c)>d_f-1$ by definition of $S^p_{d_f-1}$. Conversely, given such an $H$, all codewords of pair weight at most $d_f-1$ lie in $H$, so $\langle S^p_{d_f-1}\rangle\subseteq H\subsetneq C$, the graph $G_p^{\,d_f-1}(C)$ is disconnected, and Theorem~\ref{thm:pair-existence} applies.
    \end{proof}

    \begin{theorem}\label{thm:quotient-bound}
    Let $C$ be a linear $[n,k]_p$ symbol apir code serving as a strict $(f:d_d,d_f)_p$-FCSPC, and put $\gamma=\gamma_C^p(d_f-1)$. Then
    \[
    r_{p,}^f(k,d_d,d_f) \;\ge\; \max\Bigl\{\,N_p\bigl(q^{k},d_d\bigr),\ \ N_p\bigl(q^{\,k-\gamma},d_f\bigr)\,\Bigr\} - k .
    \]
    where $N_p(M,d)$ for the least length of a $q$-ary pair-metric code of size $M$ with minimum pair-distance at least $d$.
    \end{theorem}
    
    \begin{proof}
    The first term $N_p\bigl(q^{k},d_d\bigr)$ is the data-protection requirement as $C$ itself is a code of size $q^k$ with $d_p(C)\ge d_d$, so $n\ge N_p(q^k,d_d)$. For the second, let $H=\langle S^p_{d_f-1}\rangle$ and choose one representative from each of the $q^{k-\gamma}$ cosets of $H$ in $C$. Two representatives in distinct cosets differ by an element of $C\setminus H$, whose pair weight exceeds $d_f-1$ by definition of $S^p_{d_f-1}$. The representatives therefore form a code of size $q^{k-\gamma}$ and minimum pair-distance at least $d_f$, whence $n\ge N_p(q^{k-\gamma},d_f)$.
    \end{proof}

\section{FCSPC-DP for Specific Function Classes }\label{sec:classes}
    \subsection{Pair-locally bounded function}

    For a general function $f$, protecting $f(x)$ appears costly whenever $|\operatorname{Im}(f)|$ is large, since every pair of messages with distinct function values must be separated by at least $d_f$ in the symbol-pair metric. Messages that lie far apart, however, are already well separated, so the only pairs imposing a genuine constraint on the redundancy part are those lying within a ball of radius $d_f-1$. This shifts attention from the global size of $\operatorname{Im}(f)$ to the number of function values visible locally.

    \begin{definition}[Function pair-ball~\cite{xia2024function}]
    \label{def:function-pair-ball}
    The \emph{function pair-ball} of $f : \mathbb{F}_q^k \to \mathrm{Im}(f)$ of radius $\rho$ around $x \in \mathbb{F}_q^k$ is
    \[
        B^f_p(x,\rho) \;=\; \bigl\{\, f(x') \;:\; x' \in \mathbb{F}_q^k,\ d_p(x,x') \le \rho \,\bigr\}.
    \]
    \end{definition}

    \begin{definition}[Pair-locally bounded function]
    \label{def:pair-bounded}
    Let $\lambda \ge 2$ and $\rho \ge 1$ be integers. A function $f : \mathbb{F}_q^k \to \operatorname{Im}(f)$ is \emph{$(\lambda,\rho)$-pair-bounded} if $|B_p^f(x,\rho)| \le \lambda$ for every $x \in \mathbb{F}_q^k$. A $(2,\rho)$-pair-bounded function is also called \emph{$\rho$-pair-locally binary}~\cite{xia2024function}.
    \end{definition}

    \begin{definition}[Pair-separating colouring]
    \label{def:pair-colouring}
    A map $\operatorname{Col}_f : \mathbb{F}_q^k \to [\lambda]$ is a \emph{$\rho$-pair-separating $\lambda$-colouring} of $f$ if
    \[
    d_p(x,y) \le \rho \ \text{ and }\ f(x) \neq f(y)
    \quad\Longrightarrow\quad
    \operatorname{Col}_f(x) \neq \operatorname{Col}_f(y).
    \]
    \end{definition}

    \begin{lemma}[{\cite{11250740}}]
    \label{lem:contig-block-cond}
    Let $f$ be $(\lambda,\rho)$-pair-bounded and suppose there is a total order $\preceq$ on $\operatorname{Im}(f)$ such that $B_p^f(x,\rho)$ is a contiguous block with respect to $\preceq$ for every $x \in \mathbb{F}_q^k$. Then $f$ admits a $\rho$-pair-separating $\lambda$-colouring.
    \end{lemma}

    \begin{lemma}
    \label{lem:binary-colouring}
    Every $\rho$-pair-locally binary function admits a $\rho$-pair-separating $2$-colouring.
    \end{lemma}

    \begin{proof}
    Fix any total order $\preceq$ on $\operatorname{Im}(f)$ and set
    \[
    \operatorname{Col}_f(x) =
    \begin{cases}
        2, & \text{if } f(x) = \max_{\preceq} B_p^f(x,\rho),\\
        1, & \text{otherwise.}
    \end{cases}
    \]
    Let $d_p(x,y) \le \rho$ with $f(x) \neq f(y)$. Then $f(x), f(y) \in B_p^f(x,\rho)$, and $|B_p^f(x,\rho)| \le 2$ forces $B_p^f(x,\rho) = \{f(x),f(y)\}$. Since $d_p$ is symmetric, $d_p(y,x) \le \rho$ as well, so the same argument gives $B_p^f(y,\rho) = \{f(x),f(y)\}$. Hence $x$ and $y$ compare their function values against the \emph{same} maximum $m = \max_{\preceq}\{f(x),f(y)\}$, and exactly one of $f(x),f(y)$ equals $m$. Therefore $\operatorname{Col}_f(x) \neq \operatorname{Col}_f(y)$.
    \end{proof}

    \begin{theorem}
    \label{the:locally-red-boun}
    Let $f : \mathbb{F}_q^k \to \operatorname{Im}(f)$ admit a $(d_f-1)$-pair-separating $\lambda$-colouring, let $d_f > d_d \ge 1$, and let $C$ be a systematic linear $[n,k]_q$ code with $d_p(C) \ge d_d$. Then
    \[
    r_p^f(k, d_d, d_f) \;\le\; n - k + N_p\bigl(\lambda,\; d_f - d_d + 1\bigr),
    \]
    where $N_p(\lambda,d)$ denotes the minimum length of a symbol-pair code with $\lambda$ codewords and minimum symbol-pair distance at least $d$.
    \end{theorem}

    \begin{proof}
    Let $G$ be a generator matrix of $C$ in systematic form and $c_x = xG$. Let $\operatorname{Col}_f : \mathbb{F}_q^k \to [\lambda]$ be a $(d_f-1)$-pair-separating $\lambda$-colouring, and let $C_0 = \{\mathbf{c}_0^{(1)},\dots,\mathbf{c}_0^{(\lambda)}\}$ be a symbol-pair code of length $N_p(\lambda, d_0)$ with $\lambda$ codewords and minimum \emph{symbol-pair} distance $d_0 = d_f - d_d + 1$. Define
    \[
    C_f(x) = (c_x,\, p_x), \qquad p_x = \mathbf{c}_0^{(\operatorname{Col}_f(x))},
    \]
    so that $C_f : \mathbb{F}_q^k \to \mathbb{F}_q^{\,n + N_p(\lambda,d_0)}$ is systematic.

    \emph{Data protection.} For $x \neq y$, Lemma~\ref{lem:monotone} gives $d_p(C_f(x),C_f(y)) \ge d_p(c_x,c_y) \ge d_p(C) \ge d_d$.

    \emph{Function protection.} Let $f(x) \neq f(y)$.

    \textbf{Case 1: $d_p(x,y) \ge d_f$.} Since $C$ is systematic, two applications of Lemma~\ref{lem:monotone} give $d_p(C_f(x),C_f(y)) \ge d_p(c_x,c_y) \ge d_p(x,y) \ge d_f$.

    \textbf{Case 2: $d_p(x,y) \le d_f - 1$.} Then $\operatorname{Col}_f(x) \neq \operatorname{Col}_f(y)$ by Definition~\ref{def:pair-colouring}, so $p_x \neq p_y$ and $d_p(p_x,p_y) \ge d_0 = d_f - d_d + 1$. By Lemma~\ref{lem:symbol-pair-ineq},
    \[
    d_p\bigl(C_f(x),C_f(y)\bigr)
    \;\ge\; d_p(c_x,c_y) + d_p(p_x,p_y) - 1
    \;\ge\; d_d + (d_f - d_d + 1) - 1 \;=\; d_f .
    \]

    Hence $C_f$ is an $(f : d_d, d_f)_p$-FCSPC-DP of redundancy $n - k + N_p(\lambda, d_0)$, and the bound follows from the minimality of $r_p^f(k,d_d,d_f)$.
    \end{proof}

    Combining Lemma~\ref{lem:binary-colouring} with Theorem~\ref{the:locally-red-boun} and the evaluation $N_p(2,D) = D$, gives the following.

    \begin{corollary}
    \label{cor:binary-generic}
    Let $f$ be $(d_f-1)$-pair-locally binary, $d_f > d_d$, and let $C$ be a systematic linear $[n,k]_q$ code with $d_p(C) \ge d_d$. Then
    \[
    r_p^f(k, d_d, d_f) \;\le\; n - k + d_f - d_d + 1 .
    \]
    \end{corollary}

    The binary case admits a sharper analysis, because the two parity blocks can be chosen to differ in \emph{every} coordinate. We first record what this buys at the junctions.

    \begin{lemma}
    \label{lem:junction}
    Let $a,a' \in \mathbb{F}_q^{n}$ and $b,b' \in \mathbb{F}_q^{L}$. If $b_1 \neq b'_1$ and $b_L \neq b'_L$, then
    \[
    d_p\bigl((a,b),(a',b')\bigr) \;\ge\; d_p(a,a') + d_p(b,b') .
    \]
    \end{lemma}
    \begin{proof}
    The $n+L$ pair positions of $(a,b)$ split into the $n-1$ within-$a$ positions $(a_i,a_{i+1})$, $i < n$, the $L-1$ within-$b$ positions, and the two junctions $(a_n,b_1)$ and $(b_L,a_1)$. The differing pair positions of the standalone $a$ lie among the within-$a$ positions and the wrap $(a_n,a_1)$, so at least $d_p(a,a')-1$ within-$a$ positions differ; similarly at least $d_p(b,b')-1$ within-$b$ positions differ. Under the hypothesis both junctions differ, contributing $2$.
    \end{proof}

    \begin{proposition}
\label{lem:pair-locally-binary-fcspcdp}
Let $f : \mathbb{F}_q^k \to \operatorname{Im}(f)$ be $(d_f-1)$-pair-locally binary, let
$d_f \ge d_d + 2$, and let $C$ be a systematic linear $[n,k]_q$ code with
$d_p(C) \ge d_d + 1$. Then
\[
    r_p^f(k, d_d, d_f) \;\le\; n - k + d_f - d_d - 1 .
\]
If $d_f \le d_d + 1$, then $C$ itself is an $(f : d_d,d_f)_p$-FCSPC-DP and
$r_p^f(k,d_d,d_f) \le n-k$.
\end{proposition}
    \begin{proof}
Write $L = d_f - d_d - 1 \ge 1$, let $G$ be a systematic generator matrix of $C$ and
$c_x = xG$. Let $\operatorname{Col}_f$ be the $2$-colouring of
Lemma~\ref{lem:binary-colouring} with $\rho = d_f - 1$, and set
\[
    C_f(x) = (c_x,\, p_x), \qquad
    p_x = \begin{cases}
        1^{L}, & \text{if } \operatorname{Col}_f(x) = 2,\\
        0^{L}, & \text{if } \operatorname{Col}_f(x) = 1.
    \end{cases}
\]

\emph{Data protection.} For $x \neq y$, Lemma~\ref{lem:monotone} gives
$d_p(C_f(x),C_f(y)) \ge d_p(c_x,c_y) \ge d_d + 1 > d_d$.

\emph{Function protection.} Let $f(x) \neq f(y)$. If $d_p(x,y) \ge d_f$, then two
applications of Lemma~\ref{lem:monotone} give
$d_p(C_f(x),C_f(y)) \ge d_p(c_x,c_y) \ge d_p(x,y) \ge d_f$. Otherwise
$d_p(x,y) \le d_f-1$, so $\operatorname{Col}_f(x) \neq \operatorname{Col}_f(y)$ by
Lemma~\ref{lem:binary-colouring}, whence $\{p_x,p_y\} = \{0^L,1^L\}$. These differ in
every coordinate, so $d_p(p_x,p_y) = L$ and the hypothesis of
Lemma~\ref{lem:junction} holds. Therefore
\[
    d_p\bigl(C_f(x),C_f(y)\bigr) \;\ge\; d_p(c_x,c_y) + d_p(p_x,p_y)
    \;\ge\; (d_d + 1) + L \;=\; d_f .
\]
Hence $C_f$ is an $(f : d_d,d_f)_p$-FCSPC-DP of redundancy $n-k+L$.
\end{proof}

    \begin{lemma}[{\cite[Lem.~10]{singh2025function}}]
\label{lem:nh-asymptotic}
Let $M, D \in \mathbb{N}$ and $q \ge 2$ satisfy $D \ge 10$ and $M \le D^2$. Then
\[
    N_H(M,D) \;\le\; \frac{q(D-1)}{q\bigl(1 - \sqrt{\ln D / D}\bigr) - 1}.
\]
\end{lemma}

        \begin{corollary}
        \label{cor:locally-asymptotic}
        Let $f : \mathbb{F}_q^k \to \operatorname{Im}(f)$ admit a $(d_f-1)$-pair-separating $\lambda$-colouring, let $d_f > d_d \ge 1$, and let $C$ be a systematic linear $[n,k]_q$ code with $d_p(C) \ge d_d + 1$. In addition, suppose that
        \[
            \lambda > q, \qquad d_f - d_d \ge 2, \qquad d_f - d_d - 1 \ge 10,
            \qquad \lambda \le (d_f - d_d - 1)^2 .
        \]
        Then
        \[
            r_p^f(k,d_d,d_f)
            \;\le\; n - k + N_H\bigl(\lambda,\, d_f - d_d - 1\bigr)
            \;\le\; n - k
              + \frac{q\,(d_f - d_d - 2)}
                     {q\Bigl(1 - \sqrt{\dfrac{\ln(d_f - d_d - 1)}{d_f - d_d - 1}}\Bigr) - 1}.
        \]
        \end{corollary}

\subsection{The Symbol-Pair Weight Function}

    The \emph{symbol-pair weight function} is $f : \mathbb{F}_q^{k} \to \operatorname{Im}(f)$, $f(x) = w_p(x)$. For $k \ge 3$ its image is $\{0,2,3,\dots,k\}$: no vector has pair weight $1$, since a single non-zero coordinate already covers two cyclic pair positions.

    \begin{lemma}
    \label{lem:pair-weight-fcc}
    Let $f(x) = w_p(x)$ with $k \ge 3$, let $d_d = 2t_d+1$ and $d_f = 2t_f+1$ be odd with
    $t_f \ge t_d$, and let $C$ be a systematic linear $[n,k]_q$ code with
    $d_p(C) \ge d_d$. Then
    \[
        r_p^f(k,d_d,d_f) \;\le\; n-k+N_p\bigl(d_f,\ d_f-d_d+1\bigr).
    \]
    \end{lemma}

    \begin{proof}
    Let $G$ be a systematic generator matrix of $C$, $c_x = xG$. Let $C'$ be a symbol-pair code with $d_f$ codewords $c'_0,\dots,c'_{d_f-1}$, minimum pair-distance $d_f - d_d + 1$, and length $N_p(d_f, d_f-d_d+1)$. Define
    \[
        C_f(x) = (c_x, p_x), \qquad p_x = c'_{\,f(x) \bmod d_f} .
    \]
    
    \emph{Data protection.} For $x \neq y$, Lemma~\ref{lem:monotone} gives $d_p(C_f(x),C_f(y)) \ge d_p(c_x,c_y) \ge d_d$.
    
    \emph{Function protection.} Let $f(x) \neq f(y)$.
    
    \textbf{Case 1: $0 < |f(x)-f(y)| \le d_f - 1$.} Then $f(x) \not\equiv f(y) \pmod{d_f}$, so $p_x \neq p_y$ and $d_p(p_x,p_y) \ge d_f - d_d + 1$. By Lemma~\ref{lem:symbol-pair-ineq},
    \[
        d_p\bigl(C_f(x),C_f(y)\bigr) \;\ge\; d_p(c_x,c_y) + d_p(p_x,p_y) - 1
        \;\ge\; d_d + (d_f - d_d + 1) - 1 \;=\; d_f .
    \]
    
    \textbf{Case 2: $|w_p(x) - w_p(y)| \ge d_f$.} 
    $d_p(x,y) \ge |w_p(x) - w_p(y)| \geq d_f$. Since $C$ is systematic, $c_x = (x, z_x)$, and two applications of
    Lemma~\ref{lem:monotone} give
    $d_p(C_f(x),C_f(y)) \ge d_p(c_x,c_y) \ge d_p(x,y) \ge d_f$.
    
    Hence $C_f$ is an $(f : d_d, d_f)_p$-FCSPC-DP of redundancy
    $n - k + N_p(d_f, d_f-d_d+1)$.
    \end{proof}

\section{Extension of Classical Bounds to FCSPC}\label{sec:bounds}
    This section extends two classical bounds to the symbol-pair setting with data protection. The Plotkin-type bounds of Section~\ref{sec:plotkin} exploit the total pairwise separation forced by the two distance requirements and the sphere-packing bounds of Section~\ref{sec:sphere} exploit the disjointness of the decoding regions associated with distinct function values. 
    
    \subsection{Plotkin-Type Bounds}
    \label{sec:plotkin}

    The Plotkin bound for the $q$-ary FCBSC presented in \cite{singh2025function} can be utilised to derive the Plotkin bound for the FCSPC in the $q$-ary setting by considering the case where $b = 2$.

    \begin{lemma}[{\cite{singh2025function}}]
    \label{lem:plotkin-b}
    For any matrix $\boldsymbol{B} \in \mathbb{N}_0^{M \times M}$, writing $m = M \bmod q^{b}$,
    \[
    N_b(\boldsymbol{B}) \;\ge\;
    \frac{2q^{b}}{M^{2}(q^{b}-1) - m(q^{b}-m)}
    \sum_{i<j} [\boldsymbol{B}]_{ij}.
    \]
    \end{lemma}

    \begin{lemma}
    \label{lem:plotkin-qary}
       For any matrix $\boldsymbol{B} \in \mathbb{N}_0^{M \times M}$, and for irregular-pair distance codes over $\mathbb{F}_q$, we have
       \[
            N_p(\boldsymbol{B}) \ge \frac{2q^2}{M^2(q^2 - 1) - m(q^2 - m)} \sum_{i,j, \, i<j} [\boldsymbol{B}]_{ij},
       \]
       where $m = M \pmod{q^2}$
    \end{lemma}

    The Lemma~\ref{lem:plotkin-qary} also holds for the JPDRM (Definition~\ref{def:J-PDM}), which is defined for FCSPC with data protection.

    \begin{corollary}
    \label{cor:plotkin-jpdm}
    For any $f : \mathbb{F}_q^k \to \mathrm{Im}(f)$ and any ordering $x_1,\dots,x_{q^k}$ of $\mathbb{F}_q^{k}$, writing $\boldsymbol{D}^{(1)} = \boldsymbol{D}^{(1)}_{f}(d_d,d_f;\,x_1,\dots,x_{q^k})$ and $m = q^{k} \bmod q^{2}$,
    \[
    r_p^f(k,d_d,d_f) \;\ge\;
    \frac{2q^{2}}{q^{2k}(q^{2}-1) - m(q^{2}-m)}
    \sum_{i<j} \bigl[\boldsymbol{D}^{(1)}\bigr]_{ij}.
    \]
    \end{corollary}

    The bound in Lemma~\ref{lem:plotkin-qary} depends on the distance requirement between each pair of codewords. The next bound on $(f, d_d, d_f)$-FCSPC reduces this dependency to only the sizes of all the preimage classes and the code parameters $d_d$ and $d_f$. 


    \begin{theorem}
    \label{thm:plotkin}
    Let $f : \mathbb{F}_q^{k} \to \mathrm{Im}(f)$ with $|\mathrm{Im}(f)| \ge 2$, and set
    \[
    \Phi_f \;=\; \sum_{\alpha \in \mathrm{Im}(f)} \bigl|f^{-1}(\alpha)\bigr|^{2}.
    \]
    Then for any $d_d \le d_f$,
    \[
    r_p^f(k,d_d,d_f) \;\ge\;
    \frac{(d_f-d_d)\bigl(q^{2k}-\Phi_f\bigr) + d_d\,q^{k}(q^{k}-1)}
         {q^{2k-2}(q^{2}-1)} \;-\; k .
    \]
\end{theorem}
 
    \begin{proof}
    Let ${C_f}$ be an $(f: k,d_d,d_f)$-FCPSC-DP over $\mathbb{F}_q$ with $q^k$ codewords, each of length $n = k + r$.
    
    For a codeword $c\in C$ with corresponding message $x_c = {C_f}^{-1}(c)$ and function value $f_c = f(x_c)$, define
    \[
         S(c) = \sum_{\substack{c'\in C\\c'\neq c}} d_p(c,c').
    \]
    
    By the definition of FCPSC with data protection, $d_p(C_f)\geq d_d$  and  $d_p^f(C_f)\geq d_f$ whenever $f_{c'} \neq f_c$.  Therefore,
    \begin{equation}\label{eq:Sx}
     S(c)
    \;\geq \;
     \bigl(\rvert f^{-1}{(f_c)}\rvert - 1\bigr)d_d
    + \bigl(q^k - \rvert f^{-1}({f_c})\rvert\bigr)d_f.
    \end{equation}

    \noindent
    Summing~\eqref{eq:Sx} over all $c\in C$ and grouping by preimage class, we get the following lower bound on the pair-wise distance between codewords 
    \begin{align}
    \sum_{\substack{(c_1,c_2)\in C_f^2,\\c_1\neq c_2}}d_p(c_1,c_2) = \sum_{c\in C_f} S(c)
    &\;\geq \; \sum_{\alpha\in\Img{f}} \rvert f^{-1}({\alpha})\rvert \Bigl[\bigl\rvert(f^{-1}({\alpha})\rvert-1\bigr)d_d
        + \bigl(q^k - \rvert f^{-1}({\alpha})\rvert\bigr)d_f\Bigr]
    \notag\\
    &= d_d\sum_{\alpha\in\Img{f}}\bigl(\rvert f^{-1}({\alpha})\rvert^2 - \rvert f^{-1}({\alpha})\rvert\bigr)
    + d_f\sum_{\alpha\in\Img{f}}\rvert f^{-1}({\alpha})\rvert\bigl(q^k - \rvert f^{-1}({\alpha})\rvert\bigr)
    \notag\\
    &= d_d\bigl(\Phi_f - q^k\bigr)
    + d_f\bigl(q^{2k} - \Phi_f\bigr)\label{eq:sum-inv}
    \\
    &= (d_f-d_d)\bigl(q^{2k}-\Phi_f\bigr) - d_d\,q^k + d_f\cdot q^{2k} - d_f\,q^{2k}
    \notag\\
    &= (d_f-d_d)\bigl(q^{2k}-\Phi_f\bigr) + d_d(q^{2k}-q^k).
    \label{eq:lower}
    \end{align}
    where Equation~\ref{eq:sum-inv} follows from $\sum_{\alpha\in\Img{f}}\rvert f^{-1}({\alpha})\rvert = q^k$ and $\Phi_f = \sum_{\alpha\in\Img{f}} \rvert f^{-1}({\alpha})\rvert^2$.

    \medskip
    \noindent
    To calculate the upper bound on the sum of pairwise distances, consider a fixed coordinate position $i\in[n]$, and let $n_{ (a,b)}$ be the number of
    codewords whose $i$-th coordinate equals $(a,b)\in(\mathbb{F}_q^k)^2$.  The contribution of position~$i$ to the total pairwise distance is
    $\sum_{(a,b)\in(\mathbb{F}_q^k)^2} n_{(a,b)}^{(i)}(q^{k} - n_{(a,b)}^{(i)}) = q^{2k} - \sum_{(a,b)} (n_{(a,b)}^{(i)})^2$, which is maximized when $n_{(a,b)}^{(i)} = \frac{q^k}{q^2} = q^{k-2}$ (when the codewords are distributed as evenly as possible across the $q^2$
    pair values) for all $(a,b)$, giving $q^{2k} - q^2\cdot q^{2(k-2)} = q^{2(k-1)}(q^2-1)$.
    Summing over all $n$ coordinates:

    \begin{equation}\label{eq:upper-sp}
    \sum_{(x, y) \in {C}_f^2,\, x \neq y} d_p(x, y)
    \;\leq\; n \cdot q^{2k - 2} (q^2 - 1).
    \end{equation}
    \medskip

    From Equations \eqref{eq:lower} and \eqref{eq:upper-sp} we get,
    \[
    (d_f-d_d)\bigl(q^{2k}-\Phi_f\bigr) + d_d\,q^k(q^k-1)
    \;\leq\;
    n \cdot q^{2k - 2} (q^2 - 1).
    \]
    Substituting $n = k+r$ and rearranging:
    \[
        r
    \;\geq \;
    \frac{(d_f-d_d)\bigl(q^{2k}-\Phi_f\bigr) + d_d\,q^k(q^k-1)}
       { q^{2k - 2} (q^2 - 1)} - k.
    \]
    \end{proof}

    \noindent
    Setting $d_d = 1$ in Theorem~\ref{thm:plotkin} specialises the bound to the original FCPSC setting, where data protection is not necessary, and the systematic constraint already enforces $d_d \geq 1$.

    \begin{corollary}
    \label{cor:plotkin-fcspc}
    For any $f : \mathbb{F}_q^{k} \to \mathrm{Im}(f)$ with $|\mathrm{Im}(f)| \ge 2$ and any $d_f \ge 2$,
    \[
    r_p^f(k, d_f) \;\ge\;
    \frac{(d_f-2)\bigl(q^{2k}-\Phi_f\bigr) + 2\,q^{k}(q^{k}-1)}
         {q^{2k-2}(q^{2}-1)} \;-\; k .
    \]
    \end{corollary}

    \begin{proof}
    Every $(f : d_f)_p$-FCSPC given by a systematic encoding satisfies $d_p(C_f) \ge 2$. Hence every $(f : d_f)_p$-FCSPC is an $(f : 2, d_f)_p$-FCSPC-DP, and the claim is Theorem~\ref{thm:plotkin} with $d_d = 2$ follows.
    \end{proof}

    \subsection{Sphere-Packing Bounds}
    \label{sec:sphere}

    Throughout this subsection, $B_p(v,t) = \{y \in \mathbb{F}_q^{n} : d_p(v,y) \le t\}$ denotes the symbol-pair ball of radius $t$ about $v \in \mathbb{F}_q^{n}$, and $V_p(t,n) = |B_p(v,t)|$ its volume, which is independent of $v$ by the translation invariance of $d_p$. An $(f,t)_p$-FCSPC is a systematic encoding with $d_p^f(C_f) \ge 2t+1$.

    \begin{theorem}[Sphere-packing bound for $(f,t)_p$-FCSPC]
    \label{thm:sphere-pack-fcspc}
    Let $f : \mathbb{F}_q^k \to \mathrm{Im}(f)$ with $E = |\mathrm{Im}(f)|$ and $\ell = \min_{i \in [E]} |f^{-1}(f_i)|$. Every $(f,t)_p$-FCSPC of length $n$ satisfies
    \[
    n \;\ge\; \min\Bigl\{\, n' \,:\,
    E \cdot \min_{\substack{v_1, \dots, v_\ell \in \mathbb{F}_q^{n'}\\
    \mathrm{distinct}}}
    \Bigl|\,\bigcup_{j=1}^{\ell} B_p(v_j, t)\Bigr|
    \;\le\; q^{n'} \Bigr\},
    \]
    where $B_p(v,t)$ denotes the pair-ball of radius $t$ about $v$. Equivalently,
    \[
    r_p^f(k, t) \;\ge\; \min\Bigl\{\, n' \,:\,
    E \cdot \min_{\mathrm{distinct}}
    \Bigl|\,\bigcup_{j=1}^{\ell} B_p(v_j, t)\Bigr|
    \;\le\; q^{n'} \Bigr\} \;-\; k.
    \]
    \end{theorem}

    \begin{proof}
    Let $C$ be an $(f,t)_p$-FCSPC of length $n$. For each $i \in [E]$, consider the union of pair-balls
    \[
    U_i \;=\; \bigcup_{u \in f^{-1}(f_i)} B_p\big(C(x), t\big).
    \]
    We claim the sets $U_1, \dots, U_E$ are pairwise disjoint. Suppose that there exists some $x \in U_i \cap U_{i'}$ with $i \ne i'$. Then there are codewords $C(x)$ with $f(x) = f_i$ and $C(u')$ with $f(u') = f_{i'}$ such that $d_p(x, C(x)) \le t$ and $d_p(x, C(u')) \le t$. A received word $x$ could then be decoded to either of the codewords, yielding different function values, thereby contradicting the definition of FCSPC. Hence, the $U_i$ are pairwise disjoint.

    Each $U_i$ is a union of $|f^{-1}(f_i)|$ pair-balls of radius $t$. Because $|f^{-1}(f_i)| \ge \ell$ for every $i$, the size of each $U_i$ is at least the minimum possible size of a union of $\ell$ radius-$t$ pair-balls:
    \[
    |U_i| \;\ge\; \min_{\substack{w_1, \dots, w_\ell \in \mathbb{F}_q^n \\ \text{distinct}}}
    \left|\, \bigcup_{j=1}^{\ell} B_p(w_j, t) \,\right|
    \;=\; \left|\, \bigcup_{j=1}^{\ell} B_p(v_j, t) \,\right|,
    \]
    where $v_1, \dots, v_\ell$ achieve the minimum. As the $U_i$ are disjoint subsets of $\mathbb{F}_q^n$,
    \[
    E \cdot \left|\, \bigcup_{j=1}^{\ell} B_p(v_j, t) \,\right| \;\le\; \sum_{i=1}^{E} |U_i| \;\le\; q^n,
    \]
    Rearranging the above inequality gives the condition on $n$, and the smallest such $n$ is the stated bound.
    \end{proof}


    \begin{theorem}[Sphere-packing bound for $(f:d_d,d_f)_p$-FCSPC]
    \label{thm:sphere-pack-dp}
    Let $f : \mathbb{F}_q^k \to \mathrm{Im}(f)$ with $E = |\mathrm{Im}(f)|$ and $\ell = \min_{i \in [E]} |f^{-1}(f_i)|$. Then every $(f : d_d, d_f)_p$-FCSPC of length $n$ satisfies
    \[
    n \;\ge\; \min\Bigl\{\, n' \,:\,
    E \cdot \min_{\substack{v_1, \dots, v_\ell \in \mathbb{F}_q^{n'}\\
    d_p(v_i, v_j) \ge d_d \ \forall\, i \ne j}}
    \Bigl|\,\bigcup_{j=1}^{\ell} B_p(v_j, \lfloor (d_f - 1)/2 \rfloor)\Bigr|
    \;\le\; q^{n'} \Bigr\}.
    \]
    Equivalently,
    \[
    r_p^f(k, d_d, d_f) \;\ge\; \min\Bigl\{\, n' \,:\,
    E \cdot \min_{\substack{v_1, \dots, v_\ell\\
    d_p(v_i, v_j) \ge d_d}}
    \Bigl|\,\bigcup_{j=1}^{\ell} B_p(v_j, \lfloor (d_f - 1)/2 \rfloor)\Bigr|
    \;\le\; q^{n'} \Bigr\} \;-\; k.
    \]
    \end{theorem}

    \begin{proof}
    The proof follows the same arguments as that of Theorem~\ref{thm:sphere-pack-fcspc} with one difference. In an $(f : d_d, d_f)_p$-FCSPC, the data-protection clause guarantees $d_p(C(u_1), C(u_2)) \ge d_d$ for all $u_1 \ne u_2$. In particular, the centers $\{C(x) : u \in f^{-1}(f_i)\}$ of the balls forming each $U_i$ are pairwise at pair-distance at least $d_d$. The minimum over unconstrained configurations in the lower bound on $|U_i|$ may therefore be replaced by the minimum over configurations satisfying this separation, which is no smaller. The remainder of the argument is unchanged: the $U_i$ are pairwise disjoint because cross-class pair-distances satisfy $d_p \ge d_f \ge 2t_f + 1$, and summing over all $E$ classes gives the stated bound on $n$.
    \end{proof}

    \begin{corollary}[Explicit sphere-packing bound]\label{cor:sphere-pack-explicit}
    Let $f : \mathbb{F}_q^k \to \mathrm{Im}(f)$ with $E = |\mathrm{Im}(f)|$ and $\ell = \min_{i \in [E]} |f^{-1}(f_i)|$. Then every $(f : d_d, d_f)_p$-FCSPC of length $n$ satisfies
    \[
    r_p^f(k, d_d, d_f) \;\ge\; \min\Bigl\{\, n' \,:\,
    E \cdot \ell \cdot |B_p(\lfloor (d_d - 1)/2 \rfloor, n')| \;\le\; q^{n'} \Bigr\} \;-\; k,
    \]
    \end{corollary}

    \begin{proof}
    Since $d_p(v_i, v_j) \ge d_d$, the balls $B_p(v_j, \lfloor (d_d - 1)/2 \rfloor)$ for $j = 1, \dots, \ell$ are pairwise disjoint. As $\lfloor (d_d - 1)/2 \rfloor \le \lfloor (d_f - 1)/2 \rfloor$,
    each is contained in $B_p(v_j, \lfloor (d_f - 1)/2 \rfloor)$, so
    \[
    \Bigl|\,\bigcup_{j=1}^{\ell} B_p(v_j, \lfloor (d_f - 1)/2 \rfloor)\Bigr|
    \;\ge\;
    \Bigl|\,\bigcup_{j=1}^{\ell} B_p(v_j, \lfloor (d_d - 1)/2 \rfloor)\Bigr|
    \;=\;
    \sum_{j=1}^{\ell} |B_p(v_j, \lfloor (d_d - 1)/2 \rfloor)|
    \;=\;
    \ell \cdot |B_p(\lfloor (d_d - 1)/2 \rfloor, n')|,
    \]
    The last equality follows from the centre-independence of the pair-ball volume. Substituting into Theorem~\ref{thm:sphere-pack-dp} gives the stated bound.
    \end{proof}

    \begin{remark}
    When $f$ is a bijection, $E = q^k$, $\ell = 1$, and $d_d = d_f$, so the bound
    reduces to the sphere-packing bound for ordinary symbol-pair codes \cite[Prop.~18]{cassuto2011codes} (a pair-error-correcting code of length $n$ satisfies $q^k \cdot |B_p(t, n)| \le q^n$).
    \end{remark}

\section{Conclusion}\label{sec:conclusion}

    We introduced function-correcting symbol-pair codes with data protection, in which a prescribed level of protection is guaranteed for the message, and the function value under symbol-pair reads. Building on two distance notions, the minimum symbol-pair distance and the function minimum pair-distance, we showed that the Cassuto--Blaum relation between the Hamming and symbol-pair metrics extends to the latter, providing a two-way passage between the metrics at the level of codes.

    For optimal redundancy, we developed joint pair-distance requirement matrices and established a relationship between the length of the irregular pair-distance code corresponding to these matrices and the optimal redundancy. We also adapted the two-step construction of the Hamming setting, along with the corresponding coded pair-requirement matrices. We introduced the pair-separation constant for a function and showed that, when it is sufficiently large, data protection is free. Using the symbol-pair analogue of the $\alpha$-distance graph, we introduced the generation profile and the disconnection threshold and characterised the strict parameters realisable by a given linear code, tracing the trade-off between the strength of function protection and the number of function values protected. We gave explicit constructions for pair-locally bounded functions and the symbol-pair weight function, and extended the Plotkin and sphere-packing bounds to this setting.

\ifCLASSOPTIONcaptionsoff
  \newpage
\fi



 \bibliographystyle{abbrv}
\bibliography{bibtex/bib/main}
%

%








\end{document}